\documentclass[reqno, 12pt]{amsart}
\usepackage{amsmath, amsthm, amssymb, mathtools, graphicx, paralist, cite, color, bm}
\usepackage[all]{xypic}
\definecolor{e-mail}{rgb}{0,.40,.80}
\definecolor{reference}{rgb}{.20,.60,.22}
\definecolor{citation}{rgb}{0,.40,.80}
\usepackage[pdfstartview=FitH, colorlinks, linkcolor=reference, citecolor=citation, urlcolor=e-mail]{hyperref}

\newtheorem{thm}{Theorem}
\newtheorem{cor}[thm]{Corollary}
\newtheorem{lem}[thm]{Lemma}
\newtheorem{prop}[thm]{Proposition}

\theoremstyle{definition}
\newtheorem{defn}[thm]{Definition}
\theoremstyle{remark}
\newtheorem{rem}[thm]{Remark}
\numberwithin{thm}{section}
\theoremstyle{definition}

\theoremstyle{definition}

\theoremstyle{definition}
\numberwithin{equation}{section}

\newcommand{\N}{\mathbb N} 
\newcommand{\Z}{\mathbb Z} 

\newcommand{\K}{\mathbb K} 
\newcommand{\Kx}{\mathbb{K}(x)}

\title[Computational approach to summability via residues]{A computational approach to rational summability and its applications via discrete residues}

\author{Carlos E. Arreche}
\address{Department of Mathematical Sciences \\ The University of Texas at Dallas}
\email{arreche@utdallas.edu}

\author{Hari P. Sitaula}
\address{Department of Mathematical Sciences \\  Montana Technological University}
\email{hsitaula@mtech.edu}

\usepackage{algorithm}
\usepackage{algpseudocode}

\begin{document}

\begin{abstract}
A rational function $f(x)$ is rationally summable if there exists a rational function $g(x)$ such that $f(x)=g(x+1)-g(x)$. 
Detecting whether a given rational function is summable is an important and basic computational subproblem that arises in algorithms to study diverse aspects of shift difference equations. The discrete residues introduced by Chen and Singer in 2012 enjoy the obstruction-theoretic property that a rational function is summable if and only if all its discrete residues vanish. However, these discrete residues are defined in terms of the data in the complete partial fraction decomposition of the given rational function, which cannot be accessed computationally in general. We explain how to efficiently compute (a rational representation of) the discrete residues of any rational function, relying only on gcd computations, linear algebra, and a black box algorithm to compute the autodispersion set of the denominator polynomial. We also explain how to apply our algorithms to serial summability and creative telescoping problems, and how to apply these computations to compute Galois groups of difference equations.
\end{abstract}


\maketitle

\tableofcontents

\section{Introduction} \label{sec:introduction}

Let $\K$ be a field of characteristic zero, not assumed to be algebraically closed, and denote by $\Kx$ the field of rational functions in an indeterminate $x$ with coefficients in $\K$. Given $f(x)\in\Kx$, the \emph{rational summation problem} originally posed in \cite{Abramov:1971} asks to construct $g(x),h(x)\in\Kx$ such that \begin{equation} \label{eq:summation-problem}f(x)=g(x+1)-g(x) +h(x),\end{equation} and such that the degree of the denominator of $h(x)$ is as small as possible. Such an $h(x)$ is called a \emph{reduced form} of $f(x)$.
The rational summation problem has attracted much attention over the last half century \cite{Abramov:1971,Abramov:1975,Moenck:1977,Karr1981,Paule:1995,Pirastu1995a,Pirastu1995b, Malm:1995, Abramov1995,polyakov:2008}. It is easy to see that the problem admits a solution, starting from the trivial solution $(g(x),h(x))=(f(x-1),f(x-1))$ and applying the well-ordering principle. Moreover, given a solution $(g(x),h(x))$ we obtain another solution $(g(x)-h(x),h(x+1))$ since the degree of the denominator of $h(x+1)$ is the same as that of $h(x)$, and therefore these solutions are not unique. In the above references, several different choices are made for addressing this non-uniqueness in algorithmic settings -- for example, one can ask for the denominator of $g(x)$ to be also as small as possible, and/or to compute some (any) solution $(g(x),h(x))$ to \eqref{eq:summation-problem} as efficiently as possible. The introduction of \cite{Pirastu1995b} contains a concise summary and comparison between most of these different approaches.

Every algorithm for solving the rational summation problem also addresses, as a byproduct, the \emph{rational summability problem} of deciding, for a given $f(x)\in\Kx$, whether (just yes/no) there exists $g(x)\in \Kx$ such that $f(x)=g(x+1)-g(x)$. In this case we say that $f(x)$ is \emph{rationally summable}. There are various algorithms for addressing this simpler question, designed to forego the usually expensive and often irrelevant computation of the \emph{certificate} $g(x)$, which are presented and discussed for example in \cite{Matusevich:2000,Gerhard:2003,BCCL:2010,chen-singer:2012,Chen:2015,HouWang2015,Giesbrecht:2022} and the references therein.

The main goal of this work is to explain how to implement effectively and efficiently the residue-theoretic approach to rational summability proposed in \cite{chen-singer:2012}, following and expanding upon the preliminary presentation developed in \cite{arreche-sitaula:2024}. The \emph{discrete residues} of $f(x)\in\Kx$ (see Definition~\ref{def:dres}) are certain constants, belonging to the algebraic closure $\overline\K$ and defined in terms of the complete partial fraction decomposition of $f(x)$, which have the obstruction-theoretic property that they are all zero if and only if $f(x)$ is rationally summable. Computing these discrete residues directly from their definition is not possible in general, because neither is the apparently necessary prior computation of the complete partial fraction decomposition.

The algorithms that we propose here avoid expensive/impossible complete factorizations of denominator polynomials into linear factors, and rely \emph{almost} exclusively on gcd computations and linear algebra. We say ``almost'' because we do require a black box algorithm to compute the \emph{autodispersion set} of a polynomial, for which the most efficient known method to compute it relies on computing irreducible factorizations of denominator polynomials in $\K[x]$. Besides the method of computation, another problem inherent in the original definition of discrete residues is that in general they belong to algebraic extensions of $\K$ of arbitrarily high degree. Thus, inspired by \cite[Thm.~1]{bronstein:1993}, we propose a $\K$-rational representation of the discrete residues of $f\in\Kx$ consisting of pairs of polynomials. For each such pair, the first polynomial encodes in its set of roots \emph{where} $f$ has non-zero residues, and the second polynomial encodes in its evaluation at each such root \emph{what} the corresponding discrete residue is (see Definition~\ref{def:goal} for more details).

A key idea that enables all our computations is that iteratively applying classical Hermite reduction to $f(x)$ allows us to focus our attention on the special case where $f(x)$ has only simple poles. This idea is already suggested in \cite[\S5]{Horowitz1971}, but does not seem to be widely utilized in the literature. Thus a secondary goal of this work is to proselytize for this useful trick. Although we have not yet formally computed the theoretical worst-case complexity of this procedure (see Algorithm~\ref{alg:hermite-list}), we conjecture that it should be asymptotically the same (as a function of the degree of the denominator) as that of applying classical Hermite reduction only once.

Having reduced to the case of rational functions with only simple poles, all the varied approaches to compute reduced forms $h(x)$ of $f(x)$ as in \eqref{eq:summation-problem} seem now to proceed quite similarly to one another. As we remarked in \cite{arreche-sitaula:2024}, our own approach is based on similar ideas as that of \cite[\S5]{Gerhard:2003}. But the latter procedure is necessarily more complicated than ours, because it allows more general inputs than we need to, having carried out our earlier reduction to the case where $f(x)$ has only simple poles. The relative simplicity of our procedure is precisely what enables us to easily produce variants of our basic procedures that are better adapted to particular purposes.

Once we have computed a reduced form $\bar{f}(x)$ of $f(x)$ whose denominator is both squarefree and shift-free, the discrete residues of $f(x)$ coincide with the classical first-order residues of $\bar{f}(x)$. The factorization-free computation of the latter is achieved by the well-known computation in \cite[Lem.~5.1]{trager:1976}.

Our own interest in computing discrete residues is motivated by the following serial variant of the summability problem, which often arises as a basic subproblem in algorithms for computing Galois groups associated with shift difference equations \cite{vanderput-singer:1997,hendriks:1998}. Given an $n$-tuple $\mathbf{f}=(f_1(x),\dots,f_n(x))\in\Kx^n$, we show in Proposition~\ref{prop:v-space} how to compute (a basis for) the $\K$-vector space \begin{equation}\label{eq:multi-telescoper}V(\mathbf{f}):=\bigl\{\mathbf{v}=(v_1,\dots,v_n)\in\K^n \ | \ v_1f_1+\dots+v_nf_n \ \text{is summable}\bigr\}.\end{equation} The adaptability of our procedures enables us also to address in Proposition~\ref{prop:w-space} serial differential creative telescoping problems generalizing \eqref{eq:multi-telescoper} by replacing the $v_i\in\K$ with unknown linear differential operators with $\K$ coefficients $\mathcal{L}_i\in\K\left[\frac{d}{dx}\right]$. Also analogously, this kind of problem often arises in the computation of differential Galois groups of shift difference equations \cite{HardouinSinger2008,arreche:2017}. In the computation of this space of differential operators there arises the additional difficulty of having to produce an a priori bound on the orders of the operators that need to be considered. As far as we know, the first ever such bound is obtained in Proposition~\ref{prop:w-space-bound}, relying on our algorithmic results on discrete residues.

Our approach based on discrete residues makes it very straightforward how to accommodate several $f_i(x)$ simultaneously as in \eqref{eq:multi-telescoper}, which adaptation is less obvious (to us) how to carry out efficiently using other reduction methods (but see \cite[\S5]{Chen:2015}). There are also direct analogues of discrete residues for $q$-difference equations \cite{chen-singer:2012}, Mahler difference equations \cite{arreche-zhang:2022,arreche-zhang:2024}, and elliptic difference equations \cite{dreyfus:2018,HardouinSinger2021, babbitt:2025}. We hope and expect that the conceptual simplicity of the approach that we continue here will be useful in developing analogues in these other related (but much more technically challenging) contexts beyond the shift case. 

\section{Preliminaries} \label{sec:preliminaries}

\subsection{Basic notation and conventions} \label{sec:notation}

We denote by $\N$ the set of strictly positive integers, and by $\K$ a computable field of characteristic zero. We denote by $\overline\K$ a fixed algebraic closure of $\K$. We do not assume $\K$ is algebraically closed, and we will only refer to $\overline\K$ in proofs or for defining theoretical notions, never for computations. We shall impose the following technical assumption on $\mathbb{K}$: either it is feasible to compute integer solutions to arbitrary polynomial equations with coefficients in $\K$, or it is feasible to compute irreducible factorizations of polynomials in $\K[x]$ (see Remark~\ref{rem:shift-set}).

We denote by $\K[x]$ the ring of polynomials, and by $\Kx$ the field of rational functions, in the indeterminate $x$. For $f(x)\in \Kx$, we define \[
\sigma:f(x)\mapsto f(x+1);\qquad\text{and}\qquad \Delta:f(x)\mapsto f(x+1)-f(x).\] Note that $\sigma$ is a $\K$-linear field automorphism of $\Kx$ and $\Delta=\sigma-\mathrm{id}$ is only a $\K$-linear map with $\mathrm{ker}(\Delta)=\K$. We often suppress the functional notation and simply write $f$ instead of $f(x)$, $\sigma(f)$ instead of $f(x+1)$, etc., when no confusion is likely to arise.

We assume implicitly throughout that rational functions are normalized to have monic denominator. Even when our rational functions are obtained as (intermediate) outputs of some procedures, we will take care to arrange things so that this normalization always holds. In particular, we also assume that the outputs of $\mathrm{gcd}$ and $\mathrm{lcm}$ procedures are also always normalized to be monic. We adopt the convention that $\mathrm{deg}(0)=-1$. A \emph{proper} rational function is one whose numerator has strictly smaller degree than that of its denominator. An unadorned $\mathrm{gcd}$ or $\mathrm{lcm}$ or $\mathrm{deg}$ means that it is with respect to $x$. On the few occasions where we need a $\mathrm{gcd}$ with respect to a different variable $z$, we will write $\mathrm{gcd}_z$. We write $\frac{d}{dx}$ (resp., $\frac{d}{dz}$) for the usual derivation operator with respect to $x$ (resp., with respect to $z$).

\subsection{Partial fraction decompositions}

A polynomial $b\in\Kx$ is \emph{squarefree} if $b\neq 0$ and $\mathrm{gcd}\bigl(b,\frac{d}{dx}b\bigr)=1$. Consider a proper rational function $f=\frac{a}{b}\in\Kx$, with $\mathrm{deg}(b)\geq 1$ and $\mathrm{gcd}(a,b)=1$. Suppose further that $b$ is squarefree, and that we are given a set $b_1,\dots,b_n\in\K[x]$ of monic non-constant polynomials such that $\prod_{i=1}^nb_i=b$. Then necessarily $\mathrm{gcd}(b_i,b_j)=1$ whenever $i\neq j$, and there exist unique non-zero polynomials $a_1,\dots,a_n\in\K[x]$ with $\mathrm{deg}(a_i)<\mathrm{deg}(b_i)$ for each $i=1,\dots,n$ such that $f=\sum_{i=1}^n\frac{a_i}{b_i}$. In this situation, we denote \begin{equation}\label{eq:parfrac-def} \mathtt{ParFrac}(f;b_1,\dots,b_n):=(a_1,\dots,a_n).    
\end{equation} We emphasize that the computation of partial fraction decompositions \eqref{eq:parfrac-def} can be done very efficiently \cite{kung:1977}, provided that the denominator $b$ of $f$ has already been factored into pairwise relatively prime factors $b_i$, which need not be irreducible in $\K[x]$. Although one can similarly carry out such partial fraction decompositions more generally for pre-factored denominators $b$ that are not necessarily squarefree, our procedures only need to compute partial fraction decompositions for pre-factored squarefree denominators, in which case the notation \eqref{eq:parfrac-def} is conveniently light.

\subsection{Summability and dispersion} \label{sec:sum-disp}

We say $f\in\Kx$ is \emph{(rationally) summable} if there exists $g\in\Kx$ such that $f=\Delta(g)$. For a non-constant polynomial $b\in\K[x]$, we follow the original \cite{Abramov:1971} in defining the \emph{dispersion} of $b$ \[\label{eq:disp}\mathrm{disp}(b):=\mathrm{max}\{\ell\in\N \ | \ \mathrm{gcd}(b,\sigma^\ell(b))\neq 1\}.\] If $\mathrm{disp}(b)=0$ we say that $b$ is \emph{shiftfree}. For a rational function $f=\frac{a}{b}\in\Kx$ with $\mathrm{gcd}(a,b)=1$ and $b\notin \K$, the \emph{polar dispersion} $\mathrm{pdisp}(f):=\mathrm{disp}(b)$.

We denote by $\overline\K/\Z$ the set of orbits for the action of the additive group $\Z$ on $\overline\K$. For $\alpha\in\overline\K$, we denote \[\omega(\alpha):=\{\alpha+n \ | \ n\in\Z\},\] the unique orbit in $\overline\K/\Z$ containing $\alpha$. We will often simply write $\omega\in\overline\K/\Z$ whenever we have no need to reference a specific $\alpha\in\omega$.

\section{Discrete residues and summability}\label{sec:dres-sum}

\begin{defn}[\protect{\cite[Def.~2.3]{chen-singer:2012}}]\label{def:dres} Let $f\in\Kx$, and consider the complete partial fraction decomposition
\begin{equation} \label{eq:partial-fraction-decomposition}f=p+\sum_{k\in\N}\sum_{\alpha\in\overline\K}\frac{c_k(\alpha)}{(x-\alpha)^k},\end{equation} where $p\in\K[x]$ and all but finitely many of the $c_k(\alpha)\in\overline\K$ are zero for $k\in\N$ and $\alpha\in\overline\K$. The \emph{discrete residue} of $f$ of order $k\in\N$ at the orbit $\omega\in\overline\K/\Z$ is
\begin{equation}\label{eq:dres-def}
    \mathrm{dres}(f,\omega,k):=\sum_{\alpha\in\omega} c_k(\alpha).
\end{equation}
\end{defn}

The relevance of discrete residues to the study of rational summability is captured by the following result.

\begin{prop}[\protect{\cite[Prop.~2.5]{chen-singer:2012}}] \label{prop:dres-sum} 
$f\in\Kx$ is rationally summable if and only if $\mathrm{dres}(f,\omega,k)=0$ for every $\omega\in \overline\K$ and $k\in\N$.
    
\end{prop} As pointed out in \cite[Rem.~2.6]{chen-singer:2012}, the above Proposition~\ref{prop:dres-sum} recasts in terms of discrete residues a well-known rational summability criterion that reverberates throughout the literature, for example in \cite[p.~305]{Abramov1995}, \cite[Thm.~10]{Matusevich:2000}, \cite[Thm.~11]{Gerhard:2003}, \cite[Cor.~1]{ash:2005}, and \cite[Prop.~3.4]{HouWang2015}. All of these rely in some form or another on the following fundamental result of Abramov, which already gives an important obstruction to summability.

\begin{prop}[\protect{\cite[Prop.~3]{Abramov:1971}}] \label{prop:summable-dispersion} If a proper $0\neq f\in\Kx$ is rationally summable then $\mathrm{pdisp}(f)>0$.
\end{prop}

Discrete residues are also intimately related to the computation of reduced forms for $f$, in the following sense. As discussed in \cite[\S2.4]{chen-singer:2012}, every reduced form $h$ of $f$ as in \eqref{eq:summation-problem} has the form \begin{equation}\label{eq:general-reduced-form}h=\sum_{k\in\N}\sum_{\omega\in\overline\K/\Z}\frac{\mathrm{dres}(f,\omega,k)}{(x-\alpha_{\omega})^k}\end{equation} for some arbitrary choice of representatives $\alpha_{\omega}\in\omega$. Conversely, for any $h$ of the form \eqref{eq:general-reduced-form}, Proposition~\ref{prop:dres-sum} immediately yields that $f-h$ is rationally summable. An equivalent characterization for $h\in\Kx$ to be a reduced form is for it to have polar dispersion~$0$ (see \cite[Props.~4 \& 6]{Abramov:1975}). By \eqref{eq:general-reduced-form}, knowing the $\mathrm{dres}(f,\omega,k)$ is thus ``the same'' as knowing some/all reduced forms $h$ of $f$. But discrete residues still serve as a very useful organizing principle and technical tool, for both theoretical and practical computations.

\begin{defn}\label{def:goal} For $f\in\K(x)$, a $\K$\emph{-rational system of discrete residues} of $f$ is a sequence of pairs $(B_1,D_1),\dots,(B_m,D_m)\in\K[x]\times\K[x]$ such that:
\begin{itemize}
    \item $f$ has no poles of order higher than $m$;
    \item for each $1\leq k\leq m$, the polynomial $B_k$ is nonzero, squarefree, shiftfree, and has $\mathrm{deg}(B_k)>\mathrm{deg}(D_k)$;
    \item if $1\leq k\leq m$ and $\omega\in\overline\K/\Z$ are such that no $\alpha\in\omega$ is a root of $B_k$, then $\mathrm{dres}(f,\omega,k)=0$; and
    \item if $1\leq k\leq m$ and $\alpha\in\overline\K$ are such that $B_k(\alpha)=0$, then \[\mathrm{dres}(f,\omega(\alpha),k)=D_k(\alpha).\]
\end{itemize}
\end{defn}

Note that we do not insist on the length $m$ to be precisely equal to the highest order of a pole of $f$. The notion of $\K$-rational system of discrete residues in Definition~\ref{def:goal} was not defined explicitly in \cite{arreche-sitaula:2024}, and it is somewhat less rigid than what we set out to compute in \cite[\S3]{arreche-sitaula:2024}. We find it useful to introduce this notion explicitly, and with this added flexibility. It may not seem immediately obvious to the non-expert that there always exists a \mbox{$\K$-rational} system of discrete residues for any $f\in\Kx$. Although it is not logically necessary to formally prove their existence here (since we shall even construct them explicitly in Algorithm~\ref{alg:dres}, thus they exist!), it may perhaps be useful to offer a more conceptual Galois-theoretic justification.

\begin{lem}\label{lem:goal}
    For any $f\in\K(x)$ there exists a $\K$-rational system of discrete residues of $f$ as in Definition~\ref{def:goal}.
\end{lem}

\begin{proof}
    The Galois group $\Gamma:=\mathrm{Gal}\bigl(\overline\K(x)/\K(x)\bigr)$ is canonically identified with $\mathrm{Gal}\bigl(\overline\K/\K\bigr)$ via the restriction homomorphism $\gamma\mapsto \gamma|_{\overline\K}$. Since the action of $\Gamma$ on $\overline\K(x)$ commutes with the $\overline\K$-linear shift automorphism $\sigma:x\mapsto x+1$, we see that $\omega(\gamma(\alpha))=\gamma(\omega(\alpha))$ for every $\alpha\in\overline\K$ and $\gamma\in\Gamma$. Moreover, since $\gamma\in\Gamma$ acts trivially on $f\in\K(x)$, we see that the coefficients in the partial fraction decomposition \eqref{eq:parfrac-def} of $f$ also satisfy $c_k(\gamma(\alpha))=\gamma(c_k(\alpha))$ for every $\alpha\in\overline\K$, $k\in\N$, and $\gamma\in\Gamma$. It then follows from Definition~\ref{def:dres} that $\mathrm{dres}(f,\omega(\gamma(\alpha)),k)=\gamma\bigl(\mathrm{dres}(f,\omega(\alpha),k)\bigr)$ for each $\alpha\in\overline\K$, $k\in\N$, and $\gamma\in\Gamma$. Thus we can construct a monic, squarefree, shiftfree $B_k\in\overline\K[x]$ having precisely one root $\alpha_\omega$ in each orbit $\omega\in\overline\K/\Z$ such that \mbox{$\mathrm{dres}(f,\omega,k)\neq 0$}, and in such a way that $\alpha_{\gamma(\omega)}=\gamma(\alpha_\omega)$ for every $\gamma\in\Gamma$, which forces $B_k\in\K[x]$. There then exists a unique interpolation polynomial $D_k\in\overline\K[x]$ with $\mathrm{deg}(D_k)<\mathrm{deg}(B_k)$ such that $D_k(\alpha)=\mathrm{dres}(f,\omega(\alpha),k)$ for each root $\alpha$ of $B_k$. Since \mbox{$\gamma(D_k(\alpha))=D_k(\gamma(\alpha))$} for each such $\alpha$ and for every $\gamma\in\Gamma$, we must have $D_k\in\K[x]$.
\end{proof}

  Our goal is to compute efficiently a $\K$-rational system of discrete residues for any given $f\in\Kx$, in place of the discrete residues of $f$ as literally given in Definition~\ref{def:dres}. For such a sequence $(B_1,D_1),\dots,(B_m,D_m)$ as in our Definition~\ref{def:goal}, the polynomials $D_k$ encode all the discrete residues of $f$ of order $k$ simultaneously, since $D_k(\alpha)=\mathrm{dres}(f,\omega(\alpha),k)$ for each root $\alpha\in\overline\K$ of $B_k$, and moreover each $B_k$ contains in its set of zeros a complete and irredundant set of representatives $\alpha\in\omega$ for every orbit $\omega\in\overline\K/\Z$ such that $\mathrm{dres}(f,\omega,k)\neq 0$. Thus with such a set of data computed from $f$ we are able to write down explicitly a symbolic expression for the general theoretical reduced form $h$ given in \eqref{eq:general-reduced-form} of $f$ as in \eqref{eq:summation-problem}:
  \begin{equation}\label{eq:difference-rothstein-trager}
      h=\sum_{k=1}^m\sum_{\substack{\alpha\in \overline\K \ \text{s.t.} \\ B_k(\alpha)=0}} \frac{D_k(\alpha)}{(x-\alpha)^k}.
  \end{equation}

  \begin{rem}
     The expression in \eqref{eq:difference-rothstein-trager} is analogous to the inspirational result \cite[Thm.~1]{bronstein:1993} on symbolic partial fraction decompositions, and provides an answer to the question posed in the Remark at the end of \cite[\S6.2]{Paule:1995}, of what an analogue of the Rothstein-Trager method \cite{trager:1976} for rational integration might look like.
  \end{rem}

\section{Iterated Hermite reduction} \label{sec:hermite}

It is immediate from the Definition~\ref{def:dres} that the polynomial part $p$ of $f\in\Kx$ in \eqref{eq:partial-fraction-decomposition} is irrelevant, both for the study of summability as well as for the computation of discrete residues. So beginning with this section we restrict our attention to proper rational functions $f\in\Kx$.

Our first task is to reduce to the case where $f$ has squarefree denominator. In this section we describe how to compute $f_k\in \Kx$ for $k\in\N$ such that, relative to the theoretical partial fraction decomposition \eqref{eq:partial-fraction-decomposition} of $f$, we have \begin{equation} \label{eq:hermite-list-output-k}
    f_k=\sum_{\alpha\in\overline\K}\frac{c_k(\alpha)}{x-\alpha}.
\end{equation} Of course we will then have by Definition~\ref{def:dres}\begin{equation}
    \label{eq:hermite-list-dres-k}\mathrm{dres}(f,\omega,k)=\mathrm{dres}(f_k,\omega,1)
\end{equation} for every $\omega\in\overline\K/\Z$ and $k\in\N$.

Our computation of the $f_k\in\Kx$ satisfying \eqref{eq:hermite-list-output-k} is based on iterating classical so-called \emph{Hermite reduction} algorithms, originally developed in \cite{Ostrogradsky1845,Hermite1872} and for which we refer to the wonderful modern reference \cite[\S2.2,\S2.3]{bronstein-book}.

\begin{defn} \label{def:hermite} For proper $f\in\Kx$, the \emph{Hermite reduction} of \mbox{$f$ is} \[\mathtt{HermiteReduction}(f)=(g,h)\] where $g,h\in\Kx$ are the unique proper rational functions such that \[f=\frac{d}{dx}(g)+h\] and $h$ has squarefree denominator.    
\end{defn}

Algorithm~\ref{alg:hermite-list} computes the $f_k\in\Kx$ satisfying \eqref{eq:hermite-list-output-k} by applying Hermite reduction iteratively and scaling the intermediate outputs appropriately.

\begin{algorithm}
\caption{$\mathtt{HermiteList}$ procedure}\label{alg:hermite-list}

\begin{algorithmic}[1]

\Require A proper rational function $ 0\neq f\in\Kx$.
\Ensure A list $(f_1,\dots,f_m)$ of $f_k\in\Kx$ satisfying \eqref{eq:hermite-list-output-k}, such that $f_m\neq 0$ and $c_k(\alpha)=0$ for every $k>m$ and every $\alpha\in\overline\K$.\vspace{.1in}

\State Initialize loop: $m \gets 0$; $g \gets f$;
\While{$g \neq 0$}
\State {$(g, \hat{f}_{m+1}) \gets \mathtt{HermiteReduction}(g)$};
    \State $m \gets m+1$;
\EndWhile;
\State $f_k \gets (-1)^{k-1}(k-1)!\hat{f}_k$;\\
\Return $(f_1,\dots,f_m)$.
\end{algorithmic}
\end{algorithm}

\begin{lem} \label{lem:hermite-list}
    Algorithm~\ref{alg:hermite-list} is correct.
\end{lem}

\begin{proof}
    Let us denote by $||f||= m\in\N$ denote the highest order of any pole of $ f\in\Kx$, i.e., the largest $m\in\N$ such that there exists $\alpha\in\overline\K$ with $c_m(\alpha)\neq 0$ in \eqref{eq:parfrac-def}. Note that $f_1,\dots,f_m\in\overline\K(x)$ defined by \eqref{eq:hermite-list-output-k} are uniquely determined by having squarefree denominator and satisfying \begin{equation}\label{eq:hermite-list-condition}f=\sum_{k=1}^m\frac{(-1)^{k-1}}{(k-1)!}\frac{d^{k-1}}{dx^{k-1}}(f_k).\end{equation} Defining inductively $g_0:=f$ and \begin{gather}\notag (g_{k},\hat{f}_{k}):=\mathtt{HermiteReduction}(g_{k-1}) \\ \Longleftrightarrow\qquad g_{k-1} = \smash{\frac{d}{dx}}(g_{k})+\hat{f}_{k}\label{eq:hermite-list-induction}\end{gather} for $k\in\N$ as in Algorithm~\ref{alg:hermite-list}, we obtain by construction that all $g_k,\hat{f}_k\in\Kx$ and every $\hat{f}_k$ has squarefree denominator. Moreover, $||g_k||=||g_{k-1}||-1=m-k$, and therefore the algorithm terminates with $g_m=0$. Moreover, it follows from \eqref{eq:hermite-list-induction} that \[\sum_{k=1}^m\frac{d^{k-1}\hat{f}_k}{dx^{k-1}}=\sum_{k=1}^m\left(\frac{d^{k-1}g_{k-1}}{dx^{k-1}} - \frac{d^{k}g_k}{dx^{k}}\right)=g_0-\frac{d^mg_m}{dx^m}=f.\] Therefore the elements $(-1)^{k-1}(k-1)!\hat{f}_k$ are squarefree and satisfy \eqref{eq:hermite-list-condition}, so they agree with the $f_k\in\overline\K(x)$ satisfying \eqref{eq:hermite-list-output-k}.
\end{proof}

\begin{rem}\label{rem:hermite-list} As we mentioned in the introduction, as well as in \cite{arreche-sitaula:2024}, we do not expect Algorithm~\ref{alg:hermite-list} to be surprising to the experts. And yet it is a mystery to us why this trick is not used more widely since being originally suggested in \cite[\S5]{Horowitz1971}. We expect the theoretical cost of computing $\mathtt{HermiteList}(f)$ iteratively as in Algorithm~\ref{alg:hermite-list} to be asymptotically the same as that of computing $\mathtt{HermiteReduction}(f)$ only once, as a function of the degree of the denominator. This might seem counterintuitive, since the former is defined by applying the latter as many times as the highest order $m$ of any pole of $f$. But the size of the successive inputs in the loop decreases so quickly that the cost of the first step essentially dominates the added cost of the remaining steps put together. This conclusion is already drawn in \cite[\S5]{Horowitz1971} regarding the computational cost of computing iterated integrals of rational functions.
\end{rem}

\section{Simple Reduction}\label{sec:simple-reduction}

The results of the previous section allow us to further restrict our attention to proper rational functions $f\in\Kx$ with simple poles, which we write uniquely as $f=\frac{a}{b}$ with $a,b\in\K[x]$ such that $b$ is monic and squarefree, and such that $\mathrm{deg}(a)<\mathrm{deg}(b)$.

Our Algorithm~\ref{alg:simple-reduction} below computes a reduced form $\bar{f}\in\Kx$ such that $\mathrm{dres}(f,\omega,1)=\mathrm{dres}(\bar{f},\omega,1)$ for every $\omega\in\overline\K/\Z$ and such that $\bar{f}$ has squarefree as well as shiftfree denominator. As we mentioned already in the introduction, many algorithms have been developed, beginning with \cite{Abramov:1971}, that can compute such a reduced form, even without assuming $f$ has only simple poles. However, having earned the right to restrict our attention to the case of only simple poles by means of Algorithm~\ref{alg:hermite-list}, the Algorithm~\ref{alg:simple-reduction} seems simpler and more amenable to useful modifications (cf.~Algorithm~\ref{alg:multi-reduction}) the alternatives in the literature. But like every algorithm to compute reduced forms, Algorithm~\ref{alg:simple-reduction} still requires the computation of the following set of integers, originally defined in \cite{Abramov:1971}.

\begin{defn} \label{def:shift-set}
    For $0\neq b\in\K[x]$, the \emph{autodispersion set} of $b$ is \[\mathtt{ShiftSet}(b)=\{\ell\in\N \ | \ \mathrm{deg}(\mathrm{gcd}(b,\sigma^\ell(b)))\geq 1\}.\]
\end{defn}

 Algorithm~\ref{alg:shift-set} for computing $\mathtt{ShiftSet}(b)$ is based on a fundamental observation already made in \cite[p.~326]{Abramov:1971}, but with minor modifications to optimize the computations. Since $b$ and its squarefree form $\tilde{b}=b/\mathrm{gcd}\bigl(b,\frac{d}{dx}b\bigr)$ have the same roots, we can assume $b$ is squarefree without loss of applicability.

\begin{algorithm}
\caption{$\mathtt{ShiftSet}$ procedure}\label{alg:shift-set}
\begin{algorithmic}[1]
\Require  A squarefree polynomial $0\neq b\in\K[x]$. 
\Ensure  $\mathtt{ShiftSet}(b)$.\vspace{.1in}
\If{$\mathrm{deg}(b) \leq 1$} {$S\gets\emptyset$;}
\Else \State{$R(z) \gets \mathrm{Resultant}_x(b(x), b(x+z))$;}\vspace{.1in}
\State $\tilde{R}(z)\gets \dfrac{R(z)}{z\cdot \mathrm{gcd}_z\left(R(z), \frac{dR}{dz}(z)\right)};$ \Comment{Exact division.} \vspace{.1in}
\State $T(z)\gets \tilde{R}(z^{\frac{1}{2}})$; \Comment{
     $\tilde{R}(z)$ is even; slight speed-up.}
\State{$S \gets \{\ell\in\mathbb{N} \ | \ T(\ell^2)=0\}$;}
\EndIf;\\
\Return $S$.
\end{algorithmic}
\end{algorithm}
\begin{lem}\label{lem:shift-set}
    Algorithm~\ref{alg:shift-set} is correct.
\end{lem}

\begin{proof}
    As pointed out in \cite[p.~326]{Abramov:1971}, $\mathtt{ShiftSet}(b)$ is the set of positive integer roots of the resultant $R(z)\in\K[z]$ defined in Algorithm~\ref{alg:shift-set}, which is the same as the set of positive integer roots of the square-free part $R(z)/\mathrm{gcd}_z\bigl(R(z),\frac{dR}{dz}(z)\bigr)$. It is clear that $R(0)=0$, and since we do not care for this root, we are now looking for positive integer roots of the polynomial $\tilde{R}(z)$ defined in Algorithm~\ref{alg:shift-set}. It follows from the definition of $R(z)$ that $R(\ell)=0$ if and only if $R(-\ell)=0$ for every $\ell\in\overline\K$ (not just for $\ell\in\Z$), and we see that this property is inherited by $\tilde{R}(z)$. Since $z\nmid \tilde{R}(z)$, the even polynomial $\tilde{R}(z)=T(z^2)$ for a unique $T(z)\in \K[z]$, whose degree is evidently half of that of $\tilde{R}(z)$. 
\end{proof}

\begin{rem} \label{rem:shift-set}
    The role of the assumption that $\mathbb{K}$ be such that it is feasible to compute integer solutions to polynomial equations with $\K$-coefficients, as in \cite{Abramov:1971}, is to ensure that we can compute $\mathtt{ShiftSet}(b)$ as in Algorithm~\ref{alg:shift-set}. The role of the alternative assumption that $\K$ be such that it is feasible to compute irreducible factorizations of polynomials in $\K[x]$ is to ensure that we can apply the much more efficient (when available) algorithm of \cite{Man:1994}. We note that in \cite[\S6]{Gerhard:2003} a more efficient variant of the main ideas of \cite{Man:1994} is developed, which works for $b\in\Z[x]$.
\end{rem}

The previous Algorithm~\ref{alg:shift-set} to compute $\mathtt{ShiftSet}$ is called by the following Algorithm~\ref{alg:simple-reduction} to compute reduced forms of rational functions with squarefree denominators.

\begin{algorithm}
\caption{$\mathtt{SimpleReduction}$ procedure}\label{alg:simple-reduction}
\begin{algorithmic}[1]
\Require  A proper rational function $f\in\Kx$ with squarefree denominator.
\Ensure  A proper rational function $\bar{f}\in\Kx$, with squarefree and shiftfree denominator, such that $\mathrm{dres}(f,\omega,1)=\mathrm{dres}(\bar{f},\omega,1)$ for all \mbox{$\omega\in\overline\K/\Z$}.\vspace{.1in}
\State $b\gets \mathrm{denom}(f)$;
\State $S\gets\mathrm{ShiftSet}(b)$;
\If{$S =\emptyset $} \State {$\bar{f}\gets f$;}
\Else \For{$\ell\in S$}
\State $g_\ell\gets \mathrm{gcd}(b,\sigma^{-\ell}(b))$;
\EndFor;
\State $G\gets  \mathrm{lcm}(g_\ell \ | \ \ell\in S)$;
\State $b_0\gets b/G$; \Comment{Exact division.} 
 \For{ $\ell\in S$} 
 \State $b_\ell \gets gcd( \sigma^{-\ell}(b_0),b)$; 
\EndFor;
\State{$N\gets\{0\}\cup\{\ell\in S \ | \ \mathrm{deg}(b_\ell)\geq 1\}$}; 
\State{$(a_\ell \ | \ \ell\in N) \gets \mathtt{ParFrac}(f;b_\ell \ | \ \ell\in N);$}
\State{$\bar{f}\gets \displaystyle\sum_{\ell\in N}\sigma^\ell\left(\frac{a_\ell}{b_\ell}\right);$}
\EndIf;\\
\Return $\bar{f}$. 
\end{algorithmic}
\end{algorithm}

\begin{prop}\label{prop:simple-reduction}
Algorithm~\ref{alg:simple-reduction} is correct.
\end{prop}

\begin{proof}
    As in Algorithm~\ref{alg:simple-reduction}, let us denote by $b$ the denominator of $f$ and set $S=\mathtt{ShiftSet}(b)$. Then indeed if $S=\emptyset$, $\mathrm{pdisp}(f)=0$, so $f$ is already reduced and there is nothing to compute. Thus let us assume from now on that $S\neq\emptyset$, and let us consider roots of polynomials in $\overline\K$. For each $\ell\in S$, the roots of $g_\ell:=\mathrm{gcd}(b,\sigma^{-\ell}(b))$ are those roots $\alpha$ of $b$ such that $\alpha-\ell$ is also a root of $b$. Therefore the roots of $G=\mathrm{lcm}(g_\ell:\ell\in S)$ are those roots $\alpha$ of $b$ such that $\alpha-\ell$ is also a root of $b$ for some $\ell\in\N$ (because all possible such $\ell$ belong to $S$, by the Definition~\ref{def:shift-set} of $S$). It follows that the roots of $b_0:=b/G$ are those roots $\alpha$ of $b$ such that $\alpha-\ell$ is \emph{not} a root of $b$ for any $\ell\in\N$. In particular, $\mathrm{disp}(b_0)=0$. We call $b_0$ the \emph{divisor of initial roots}.
    
    Now the roots of $b_\ell:=\mathrm{gcd}(\sigma^{-\ell}(b_0),b)$ are those roots $\alpha$ of $b$ such that $\alpha-\ell$ is a root of $b_0$, i.e., the roots of $b$ which are precisely $\ell$ shifts away from the initial root in their respective $\mathbb{Z}$-orbits. It may happen that $b_\ell=1$ for some $\ell\in S$, because even though each $\ell\in S$ is the difference between two roots of $b$, it might be that no such pair of roots of $b$ involves any initial roots of $b_0$. Writing $N:=\{0\}\cup \{\ell\in S \ | \ \mathrm{deg}(b_\ell)\geq 1\}$, it is clear that \begin{equation}\label{eq:simple-reduction-factorization}\prod_{\ell\in N}b_\ell=b\qquad \text{and} \qquad \mathrm{gcd}(b_\ell,b_j)=1 \quad \text{for}\quad \ell\neq j.\end{equation} Therefore we may uniquely decompose $f$ into partial fractions as in \eqref{eq:parfrac-def} with respect to the factorization \eqref{eq:simple-reduction-factorization}, as called by Algorithm~\ref{alg:simple-reduction}: \[f=\sum_{\ell\in N} \frac{a_\ell}{b_\ell},\qquad \text{and set} \qquad \bar{f}:= \displaystyle\sum_{\ell\in N}\sigma^\ell\left(\frac{a_\ell}{b_\ell}\right).\]
    
    Now this $\bar{f}$ is a sum of proper rational functions with squarefree denominators, whence $\bar{f}$ also is proper with squarefree denominator.
    Since $\sigma^\ell(b_\ell)=\mathrm{gcd}(b_0,\sigma^\ell(b))$ is a factor of $b_0$ for each $\ell\in N$ and $\mathrm{disp}(b_0)=0$, we conclude that $\mathrm{pdisp}(\bar{f})=0$. Finally, for each $\ell\in N-\{0\}$ we see that \[\frac{a_\ell}{b_\ell}-\sigma^\ell\left(\frac{a_\ell}{b_\ell}\right)=\sum_{i=0}^{\ell-1}\sigma^i\left(\frac{a_\ell}{b_\ell}-\sigma\left(\frac{a_\ell}{b_\ell}\right)\right),\] whence $f-\bar{f}$ is a sum of rationally summable elements, and is therefore itself rationally summable. By Proposition~\ref{prop:dres-sum}, this implies $\mathrm{dres}(f-\bar{f},\omega,k)=0$ for every $\omega\in\overline\K/\Z$ and every $k\in\N$.
\end{proof}

\begin{rem}\label{rem:simple-reduction}
  As we stated in the introduction, Algorithm~\ref{alg:simple-reduction} strikes us as being conceptually similar to the one already developed in \cite[\S5]{Gerhard:2003}, but its description is made simpler by our restriction to rational functions with simple poles only. Although having a procedure that is easier for humans to read is not always computationally advantageous, it is so in this case, because the relative simplicity of Algorithm~\ref{alg:simple-reduction} will enable us in \S\ref{sec:extensions} to easily develop modifications that extend its reach beyond just rational summability.  
\end{rem}

\section{Computation of discrete residues}

Now we wish to put together the algorithms presented in the earlier sections to compute symbolically all the discrete residues of an arbitrary proper $f\in\Kx$, in the sense described in \S\ref{sec:dres-sum}. In order to do this, we first recall the following result describing the sense in which we compute classical residues symbolically by means of an auxiliary polynomial, and its short proof which explains how to actually compute this polynomial in practice.

\begin{lem}[\protect{\cite[Lem.~5.1]{trager:1976}}]\label{lem:trager}
    Let $f=\frac{a}{b}\in\Kx$ such that $a,b\in\K[x]$ satisfy $a\neq 0$, $\mathrm{deg}(a)<\mathrm{deg}(b)$, $\mathrm{gcd}(a,b)=1$, and $b$ is squarefree. Then there exists a unique polynomial $0\neq r\in\K[x]$ such that $\mathrm{deg}(r)<\mathrm{deg}(b)$ and \[f=\sum_{\{\alpha\in\overline\K \ | \ b(\alpha)=0\}}\frac{r(\alpha)}{x-\alpha}.\]
\end{lem}

\begin{proof}
    Since the set of poles of $f$ (all simple poles) is the set of roots of $b$, we know the classical first-order residue $c_1(\alpha)$ of $f$ at each $\alpha\in\overline\K$ such that $b(\alpha)=0$ satisfies $0\neq c_1(\alpha)=a(\alpha)/\frac{db}{dx}(\alpha).$ Using the extended Euclidean algorithm we can compute the unique $0\neq r$ in $\K[x]$ such that $\mathrm{deg}(r)<\mathrm{deg}(b)$ and $r\cdot\frac{d}{dx}(b)\equiv a \pmod{b}.$
    \end{proof}

For $f\in\Kx$ satisfying the hypotheses of Lemma~\ref{lem:trager}, we denote \begin{equation}\label{eq:first-residues-def}\mathtt{FirstResidues}(f):=(b,r),\end{equation} where $b,r\in\K[x]$ are also as in the notation of Lemma~\ref{lem:trager}. We also define $\mathtt{FirstResidues}(0):=(1,0)$, for convenience. With this, we can now describe the following simple Algorithm~\ref{alg:dres} to compute a symbolic representation of the discrete residues of $f$.

\begin{algorithm}
\caption{$\mathtt{DiscreteResidues}$ procedure}\label{alg:dres}
\begin{algorithmic}[1]
\Require  A proper rational function $0\neq f\in\Kx$.
\Ensure $\bigl((B_1,D_1),\dots,(B_m,D_m) \bigr)$, a $\K$-rational system of discrete residues of $f$ as in Definition~\ref{def:goal}. \vspace{.1in}

\State $(f_1,\dots,f_m)\gets \mathtt{HermiteList}(f)$;
\For{$k=1..m$}
\State $\bar{f}_k\gets\mathtt{SimpleReduction}(f_k)$;
\State $(B_k,D_k)\gets \mathtt{FirstResidues}(\bar{f}_k)$;
\EndFor;\\
\Return $\bigl((B_1,D_1),\dots,(B_m,D_m) \bigr)$. 
\end{algorithmic}
\end{algorithm}

\begin{thm}    
\label{thm:dres}
    Algorithm~\ref{alg:dres} is correct.
\end{thm}

\begin{proof}
    It follows from the correctness of Algorithm~\ref{alg:hermite-list} proved in Lemma~\ref{lem:hermite-list} that $f$ has no poles of order greater than $m$, whence by Definition~\ref{def:dres} every non-zero discrete residue of $f$ has order at most $m$. Consider now $\bar{f}_k:=\mathtt{SimpleReduction}(f_k)$, which by the correctness of Algorithm~\ref{alg:simple-reduction} is such that $f_k-\bar{f}_k$ is summable. We prove the correctness of Algorithm~\ref{alg:dres} for each $k=1,\dots,m$ depending on whether $\bar{f}_k=0$ or not.

    In case $\bar{f}_k=0$, Algorithm~\ref{alg:dres} produces $(B_k,D_k)=(1,0)$. In this case we also know that $f_k$ is summable, and therefore by \eqref{eq:hermite-list-dres-k} $\mathrm{dres}(f_k,\omega,1)=\mathrm{dres}(f,\omega,k)=0$ for every $\omega\in\overline\K/\Z$. Thus the output of Algorithm~\ref{alg:dres} is correct in this case.

    Suppose now that $\bar{f}_k\neq 0$. It follows from the definition of $(B_k,D_k):=\mathtt{FirstResidues}(\bar{f}_k)$ as in \eqref{eq:first-residues-def} that $B_k$ is the denominator of the proper rational function $\bar{f}_k$, and therefore $B_k$ is non-constant, squarefree, and has $\mathrm{disp}(B_k)=0$, by the correctness of Algorithm~\ref{alg:simple-reduction} proved in Lemma~\ref{prop:simple-reduction}. Let us denote by $\bar{c}_k(\alpha)$ the classical first order residue of $\bar{f}_k$ at each $\alpha\in\overline\K$ (note that $\bar{f}_k$ has only simple poles, so there are no other residues). We obtain from Lemma~\ref{lem:trager} that $D_k\neq 0$, $\mathrm{deg}(D_k)<\mathrm{deg}(B_k)$, and $D_k(\alpha)=\bar{c}_k(\alpha)$ for each root $\alpha$ of $B_k$. Since $\bar{f}_k$ has at most one pole in each orbit $\omega\in\overline\K/\Z$ (this is the very meaning of $\mathrm{pdisp}(\bar{f}_k)=0$), it follows that $\bar{c}_k(\alpha)=\mathrm{dres}(\bar{f}_k,\omega(\alpha),1)$ for every $\alpha\in\overline\K$. To conclude, we observe that \[\mathrm{dres}(\bar{f}_k,\omega,1)=\mathrm{dres}(f_k,\omega,1)=\mathrm{dres}(f,\omega,k)\] for each $\omega\in\overline\K/\Z$; where the first equality follows from the summability of \mbox{$f_k-\bar{f}_k$} established in the proof of Proposition~\ref{prop:simple-reduction}, and the second equality is \eqref{eq:hermite-list-dres-k}.\end{proof}

\begin{rem}\label{rem:dres-deficiency}
    As we mentioned in \S\ref{sec:dres-sum}, the knowledge of a reduced form $h$ for $f$ is morally ``the same'' as knowledge of the discrete residues of $f$. And yet, the output $\bigl((B_1,D_1),\dots,(B_m,D_m) \bigr)$ of Algorithm~\ref{alg:dres} has the following deficiency: it may happen that for some $j\neq k$, we have $\mathrm{dres}(f,\omega,k)\neq 0 \neq \mathrm{dres}(f,\omega,j)$, and yet the representatives $\alpha_j,\alpha_k\in\omega$ such that $B_j(\alpha_j)=0=B_k(\alpha_k)$ may be distinct, with $\alpha_j\neq\alpha_k$. In many applications, this is not an issue because summability problems decompose into parallel summability problems in each degree component, as we see from Proposition~\ref{prop:dres-sum}. Actually, the systematic exploitation of this particularity was one of the original motivations behind Algorithm~\ref{alg:hermite-list}. But it is still somewhat unsatisfying that the different $B_k$ associated to the same $f$ are not better coordinated. Moreover, this does become a more serious issue in further applications to creative telescoping, where the discrete residues of different orders begin to interact. We explain how to address this problem in Remark~\ref{rem:dres-deficiency-fix}, when we have developed the requisite technology.
\end{rem}

\section{Generalizations and applications to serial summability and creative telescoping}\label{sec:extensions}

In this section we collect some modifications to the procedures described in the previous sections to produce outputs that allow for more immediate comparison of discrete residues across several rational functions and across different orders.

\subsection{Serial summability}\label{sec:serial-summability}
We begin with the parameterized summability problem \eqref{eq:multi-telescoper} described in the introduction. Let $\mathbf{f}=(f_1,\dots,f_n)\in\Kx^n$ be given, and suppose we wish to compute a $\mathbb{K}$-basis for \begin{equation}\label{eq:v-space}V(\mathbf{f}):=\left\{(v_1,\dots,v_n)\in\K^n \ \middle| \ \sum_{i=1}^n v_i\cdot f_i \quad \text{is summable}\right\}.\end{equation}  It follows from Definition~\ref{def:dres} that, for every $(v_1,\dots,v_n)\in\K^n$, $\omega\in\overline\K/\Z$, and $k\in \N$, we have \begin{equation}\label{eq:dres-linearity}\mathrm{dres}\left(\sum_{i=1}^nv_i\cdot f_i,\,\omega,\,k\right)=\sum_{i=1}^nv_i\cdot\mathrm{dres}(f_i,\omega,k).\end{equation} Therefore, by Proposition~\ref{prop:dres-sum}, we obtain the alternative characterization
\begin{equation}\label{eq:v-space-test}\begin{gathered}(v_1,\dots,v_n)\in V(\mathbf{f}) \qquad\Longleftrightarrow\\
\sum_{i=1}^nv_i\cdot\mathrm{dres}(f_i,\omega,k)=0\quad \text{for every} \quad \omega\in\overline\K/\Z \quad \text{and every} \quad k\in\N.\end{gathered}\end{equation}
This is a linear system that we can solve for the unknown vectors in $ V(\mathbf{f})$ as soon as we know how to write it down explicitly. The problem is that applying Algorithm~\ref{alg:dres} to each $f_i$ separately yields \[\mathtt{DiscreteResidues}(f_i)=\bigl((B_{i,1},D_{i,1}),\dots,(B_{i,m_i},D_{i,m_i})).\] It may happen that many different $\alpha_{i,k}\in\overline\K$ all belong to the same orbit $\omega\in\K/\Z$ and satisfy $B_{i,k}(\alpha_{i,k})=0$. In this case, the condition \eqref{eq:v-space-test} is \emph{not} equivalent to asking for $\sum_i v_i \cdot D_{i,k}=0$ for each $k\in\N$ (but see Proposition~\ref{prop:v-space} below). This leads to the undesirable bookkeping problem of having to decide, for each fixed $k\leq\max\{m_i \ | \ i=1,\dots,n\}$, for which orbits $\omega\in\overline\K/\Z$ it might happen that we have several different $\alpha_i\in\omega$ which are roots of $B_{i,k}$ but are not equal to one another on the nose.

To address this kind of problem, we introduce in Algorithm~\ref{alg:multi-reduction} a generalization of Algorithm~\ref{alg:simple-reduction} that computes reduced forms $\bar{f}_1,\dots,\bar{f}_n$ for several $f_1,\dots,f_n$ simultaneously and \emph{compatibly}, in the sense that if $\alpha_i\in\overline\K$ is a pole of $\bar{f}_i$ and $\alpha_j\in\overline\K$ is a pole of $\bar{f}_j$ such that $\omega(\alpha_i)=\omega(\alpha_j)$ then we insist on having $\alpha_i=\alpha_j$ on the nose. Note that this compatibility is succintly captured by the requirement that the product of the denominators of the $\bar{f}_i$ be shiftfree, which is how this compatibility is expressed in the Output of Algorithm~\ref{alg:multi-reduction} below. As in Algorithm~\ref{alg:simple-reduction}, we may assume that all the $f_i$ are proper and, thanks to Algorithm~\ref{alg:hermite-list}, that they all have squarefree denominators.

\begin{algorithm}
\caption{$\mathtt{SimpleReduction}^+$ procedure}\label{alg:multi-reduction}
\begin{algorithmic}[1]
\Require  An $n$-tuple $(f_1,\dots,f_n)\in\Kx^n$ of proper rational functions with squarefree denominators.
\Ensure  An $n$-tuple $(\bar{f}_1,\dots,\bar{f}_n)\in\Kx^n$ of proper rational functions with squarefree denominators, such that each $\mathrm{dres}(f_i,\omega,1)=\mathrm{dres}(\bar{f}_i,\omega,1)$ for every $\omega\in\overline\K/\Z$, and such that the product of the denominators of the $\bar{f}_i$ is shiftfree.\vspace{.1in}

 \State $(b_1,\dots,b_n)\gets \bigl(\mathrm{denom}(f_1),\dots,\mathrm{denom}(f_n)\bigr)$;
\State $b\gets\mathrm{lcm}(b_1,\dots,b_n)$
\State $S\gets\mathrm{ShiftSet}(b)$;
\If{$S =\emptyset $} \State {$(\bar{f}_1,\dots,\bar{f}_n)\gets (f_1,\dots,f_n)$;}
\Else \For{$\ell\in S$}
\State $g_\ell\gets \mathrm{gcd}(b,\sigma^{-\ell}(b))$;
\EndFor;
\State $G\gets  \mathrm{lcm}(g_\ell \ | \ \ell\in S)$;
\State $b_0\gets \frac{b}{G}$; \Comment{Exact division.}
\For{ $i=1..n$}
 \For{ $\ell\in S\cup \{0\}$} 
 \State $b_{i,\ell} \gets \mathrm{gcd}( \sigma^{-\ell}(b_0),b_i)$; 
\EndFor;
\State{$N_i\gets\{\ell\in S\cup\{0\} \ \big| \ \mathrm{deg}(b_{i,\ell})\geq 1\bigr\}$}; \Comment{Remove $b_{i,\ell}$ if $b_{i,\ell}=1$.}
\State $(a_{i,\ell} \ | \ \ell\in N_i) \gets \mathtt{ParFrac}(f_i;b_{i,\ell} \ | \ \ell\in N_i);$ 
\State{$\bar{f}_i\gets \displaystyle\sum_{\ell\in N_i}\sigma^\ell\left(\frac{a_{i,\ell}}{b_{i,\ell}}\right);$}
\EndFor;\EndIf;\\
\Return $(\bar{f}_1,\dots,\bar{f}_n)$. 
\end{algorithmic}
\end{algorithm}

\begin{cor}\label{cor:multi-reduction}
    Algorithm~\ref{alg:multi-reduction} is correct.
\end{cor}

\begin{proof}
The proof is very similar to that of Proposition~\ref{prop:simple-reduction}, so we only sketch the main points. The key difference is that now $b_0$ has been defined so that for each root $\alpha$ of $b_0$ and each root $\alpha_i$ of $b_i$ belonging to $\omega(\alpha)$ we have that $\alpha_i-\alpha\in\mathbb{Z}_{\geq 0}$. The roots of $b_{i,\ell}$ are precisely those roots of $b_i$ which are $\ell$ steps away from the unique root of $b_0$ that belongs to the same orbit. By construction, the denominator of each $\bar{f}_i$ is a factor of $b_0$, which has $\mathrm{disp}(b_0)=0$ as before. \end{proof}

\begin{rem}\label{rem:dres-deficiency-fix}
     We are now able to address the deficiency discussed in Remark~\ref{rem:dres-deficiency} with the aid of Algorithm~\ref{alg:multi-reduction}. Indeed, for a non-zero proper $f\in\Kx$, let us define $(f_1,\dots,f_m):=\mathtt{HermiteList}(f)$ as in Algorithm~\ref{alg:dres}. If we now set $(\bar{f}_1,\dots,\bar{f}_m):=\mathtt{SimpleReduction}^+(f_1,\dots,f_m)$, instead of setting $\bar{f}_k:=\mathtt{SimpleReduction}(f_k)$ separately for $k=1,\dots,m$, we will no longer have the problem of the $B_k$ being incompatible.
\end{rem}

More generally, we can combine Algorithm~\ref{alg:multi-reduction} with the modification proposed in the above Remark~\ref{rem:dres-deficiency-fix} to compute symbolic representations of the discrete residues $\mathrm{dres}(f_i,\omega,k)$ of several $f_1,\dots,f_n\in\Kx$, which are compatible simultaneously across the different $f_i$ as well as across the different $k\in\N$. This will be done in Algorithm~\ref{alg:multi-dres}, after explaining the following small necessary modification to the $\mathtt{FirstResidues}$ procedure defined in \eqref{eq:first-residues-def}. For an $n$-tuple of proper rational functions $\mathbf{f}=(f_1,\dots,f_n)$ with squarefree denominators, suppose $\mathtt{FirstResidues}(f_i)=:(b_i,r_i)$ as in \eqref{eq:first-residues-def}, and let $b:=\mathrm{lcm}(b_1,\dots,b_n)$. Letting $a_i:=\mathrm{numer}(f_i)$ and $d_i:=\frac{b}{b_i}$, we see that $\mathrm{gcd}(b_i,d_i)=1$ because $b$ is squarefree, and therefore by the Chinese Remainder Theorem we can find a unique $p_i\in\K[x]$ with $\mathrm{deg}(p_i)<b$ such that \[p_i\cdot\frac{d}{dx}(b_i)\equiv a_i \pmod{ b_i }\qquad \text{and} \qquad p_i\equiv 0 \pmod{d_i}.\] Then we see that $p_i(\alpha)$ is the first-order residue of $f_i$ at each root $\alpha$ of $b$, whether or not $f_i$ actually has a pole at such an $\alpha$. With notation as above, define \[\mathtt{FirstResidues}^+(\mathbf{f}):=(b;(p_1,\dots,p_n)).\] 

\begin{algorithm}
\caption{$\mathtt{DiscreteResidues}^+$ procedure}\label{alg:multi-dres}
\begin{algorithmic}[1]
\Require  An $n$-tuple $(f_1,\dots,f_n)\in\Kx^n$ of proper non-zero rational functions.
\Ensure A pair $(B;\mathbf{D})$, consisting of a polynomial $B\in\K[x]$, 
and an array $\mathbf{D}=\bigl(D_{i,k} \ | \ 1\leq i \leq n; 1\leq k \leq m \bigr)$ of polynomials $D_{i,k}\in\K[x]$, such that $(B,D_{i,1}),\dots,(B,D_{i,m})$ is a $\K$-rational system of discrete residues of $f_i$ as in Definition~\ref{def:goal} for each $1\leq i\leq n$.\vspace{.1in}

\For{$i=1..n$}
\State $(f_{i,1},\dots,f_{i,m_i})\gets \mathtt{HermiteList}(f_i)$;
\EndFor;
\State $m\gets \mathrm{max}\{m_1,\dots,m_n\}$;
\ForAll{$i=1..n$ such that $m_i<m$}
\For{$k=(m_i+1)..m$}
\State $f_{i,k}\gets 0$;\Comment{Extend short vectors by $0$.}
\EndFor;
\EndFor;
\State $\mathbf{f}\gets (f_{i,k} \ | \ 1\leq i \leq n; 1\leq k \leq m)$
\State $\bar{\mathbf{f}}\gets\mathtt{SimpleReduction}^+(\mathbf{f})$;\\
\Return $\mathtt{FirstResidues}^+(\bar{\mathbf{f}})$. 
\end{algorithmic}
\end{algorithm}

\begin{cor}\label{cor:multi-dres}
    Algorithm~\ref{alg:multi-dres} is correct.
\end{cor}
\begin{proof}
    This is an immediate consequence of the correctness of Algorithm~\ref{alg:multi-reduction} proved in Corollary~\ref{cor:multi-reduction}, coupled with the same proof, \emph{mutatis mutandis}, given for Theorem~\ref{thm:dres}.
\end{proof}

Algorithm~\ref{alg:multi-dres} leads immediately to a simple algorithmic solution of the problem of computing $V(\mathbf{f})$ in \eqref{eq:v-space}.

\begin{prop}\label{prop:v-space}
    Let $\mathbf{f}=(f_1,\dots,f_n)$ with each $0\neq f_i\in\Kx$ proper. Let  \[\mathtt{DiscreteResidues}^+(\mathbf{f})=(B,\mathbf{D}),\qquad \text{where} \qquad \mathbf{D}=(D_{i,k} \ | \ 1\leq i\leq n; \ 1\leq k\leq m)\] is as in Algorithm~\ref{alg:multi-dres}, and let $V(\mathbf{f})$ be as in \eqref{eq:v-space}. Then \begin{equation}\label{eq:v-space-proof}V(\mathbf{f})=\left\{(v_1,\dots,v_n)\in\K^n \ \middle| \ \sum_{i=1}^n v_i\cdot D_{i,k}=0 \ \text{for each} \ 1\leq k \leq m\right\}.\end{equation}
\end{prop}

\begin{proof}
For each $(v_1,\dots,v_n)\in\K^n$ and $\alpha\in\overline{\K}$ such that $B(\alpha)=0$, \begin{equation}\label{eq:dres-application}\mathrm{dres}\left(\sum_{i=1}^nv_i f_i,\omega(\alpha),k\right)=\sum_{i=1}^nv_iD_{i,k}(\alpha)\end{equation} by the correctness of Algorithm~\ref{alg:multi-dres} proved in Corollary~\ref{cor:multi-dres}. Moreover, since $\mathrm{deg}(D_{i,k})<\mathrm{deg}(B)$ (whenever $D_{i,k}\neq 0$) and $B$ is squarefree, we see that for each given $1\leq k\leq m$ the expression \eqref{eq:dres-application} is zero for every root $\alpha$ of $B$ if and only if the polynomial $\sum_iv_iD_{i,k}\equiv 0$ identically. We conclude by Proposition~\ref{prop:dres-sum}.\end{proof}

\subsection{Serial creative telescoping}\label{sec:serial-telescoping}
In \cite[Remark~7.5]{arreche-sitaula:2024} we stated that one could produce without too much additional effort a variant of Proposition~\ref{prop:v-space} that computes more generally the $\K$-vector space of solutions to the creative telescoping problem obtained by replacing the unknown coefficients $v_i\in\K$ in \eqref{eq:multi-telescoper} with unknown linear differential operators $\mathcal{L}_i\in\K\bigl[\frac{d}{dx}\bigr]$. We shall produce such a variant here.

For an $n$-tuple $\mathbf{f}=(f_1,\dots,f_n)\in\Kx^n$ as before, let us denote by \begin{equation}\label{eq:w-space}W(\mathbf{f}):=\left\{\bigl(\mathcal{L}_1,\dots,\mathcal{L}_n\bigr) \in\K\left[\frac{d}{dx}\right]^n \ \middle| \ \sum_{i=1}^n\mathcal{L}_i(f_i) \ \text{is summable}\right\}.\end{equation} Note that this is a (left) module over the commutative principal ideal domain $\K\left[\frac{d}{dx}\right]$, which is abstractly isomorphic to the polynomial ring $\K[X]$. Our goal is to compute a $\K\left[\frac{d}{dx}\right]$-basis for the free $\K\left[\frac{d}{dx}\right]$-module of finite rank $W(\mathbf{f})$. 

We begin with a computational result.
\begin{lem}\label{lem:w-space} For $f\in\K(x)$, $\omega\in\overline\K/\Z$, $k\in\N$, and a linear differential operator $\mathcal{L}=\sum_j\lambda_j\cdot\frac{d^j}{dx^j}\in\K\left[\frac{d}{dx}\right]$, the discrete residue \[\mathrm{dres}\bigl(\mathcal{L}(f),\omega,k\bigr)=\sum_{j=0}^{k-1}\lambda_j\frac{(-1)^j(k-1)!}{(k-1-j)!}\mathrm{dres}(f,\omega,k-j).\]
\end{lem}

\begin{proof}
    This follows from the elementary computation that \[\mathcal{L}\left( \sum_{k\in\N}\sum_{\alpha\in\overline\K}\frac{c_k(\alpha)}{(x-\alpha)^k}\right)=\sum_{k\in\N}\sum_{\alpha\in\overline\K}\left(\sum_{j=0}^{k-1}\lambda_j\frac{(-1)^j(k-1)!}{(k-1-j)!}c_{k-j}(\alpha)\right)\frac{1}{(x-\alpha)^k}\] and the Definition~\ref{def:dres}.
\end{proof}

The following differential analogue of \eqref{eq:dres-linearity} is an immediate consequence of Lemma~\ref{lem:w-space} and the $\K$-linearity of discrete residues evident from \eqref{eq:dres-def}.

\begin{cor}\label{cor:w-space} For $f_1,\dots,f_n\in\K(x)$, $\omega\in\overline\K/\Z$, $k\in\N$, and linear differential operators $\mathcal{L}_i=\sum_j\lambda_{i,j}\cdot\frac{d^j}{dx^j}\in\K\left[\frac{d}{dx}\right]$, the discrete residue 
\[
\mathrm{dres}\left(\sum_{i=1}^n\mathcal{L}_i(f_i),\,\omega,\,k\right)=\sum_{i=1}^n\sum_{j=0}^{k-1}\lambda_{i,j}\frac{(-1)^j(k-1)!}{(k-1-j)!}\mathrm{dres}(f_i,\omega,k-j).
\]
\end{cor}

Let us introduce, for convenience, the following bounded-order variant of $W(\mathbf{f})$ in \eqref{eq:w-space}. For a bound $\beta\in\Z_{\geq 0}$, we denote by \begin{equation}\label{eq:w-space-bounded}
    W^\beta(\mathbf{f}):=\bigl\{(\mathcal{L}_1,\dots,\mathcal{L}_n)\in W(\mathbf{f}) \ \big| \ \mathrm{ord}(\mathcal{L}_i)\leq \beta \ \text{for} \ 1\leq i \leq n\bigr\}.
\end{equation}

\begin{prop}\label{prop:w-space}
    Let $\mathbf{f}=(f_1,\dots,f_n)$ with each $0\neq f_i\in\Kx$ proper. Let  \[\mathtt{DiscreteResidues}^+(\mathbf{f})=(B,\mathbf{D}),\qquad \text{where} \qquad \mathbf{D}=(D_{i,k} \ | \ 1\leq i\leq n; \ 1\leq k\leq m)\] is as in Algorithm~\ref{alg:multi-dres}, and let $W^\beta(\mathbf{f})$ be as in \eqref{eq:w-space-bounded}. Then for each $\beta\in\Z_{\geq 0}$, \begin{equation}\label{eq:w-space-proof}W^\beta(\mathbf{f})\cong\left\{\ \begin{pmatrix}
        \lambda_{1,0} & \cdots & \lambda_{1,\beta} \\
        \vdots & & \vdots \\
        \lambda_{n,0} & \cdots & \lambda_{n,\beta}
    \end{pmatrix}\in\K^{n\times (\beta+1)} \ \middle| \ \begin{gathered}\sum_{i=1}^n \sum_{j=1}^m\vphantom{\begin{matrix}
        X \\ X \\ X 
    \end{matrix}}\lambda_{i,k-j}\frac{(-1)^{j}D_{i,j}}{(j-1)!} =0
    \\ \text{for each} \ 1\leq k \leq m+\beta\vphantom{\frac{A}{B}}\end{gathered}\right\},\end{equation} where the isomorphism is given explicitly by writing $\mathcal{L}_i=\sum_{j=0}^\beta \lambda_{i,j}\frac{d^j}{dx^j}$, and with the convention that $\lambda_{i,k-j}:=0$ whenever $j>k$.
\end{prop}

\begin{proof}
    Since the order of each $\mathcal{L}_i\in\K\left[\frac{d}{dx}\right]$ is at most $\beta$ and the highest order of any pole of each $f_i$ is at most $m$, it follows from the Definition~\ref{def:goal} that every $\mathrm{dres}(\mathcal{L}_i(f_i),\omega,k)=0$ for every $\omega\in\overline\K/\Z$ and every $k>m+\beta$. For the remaining $1\leq k \leq m+\beta$, re-indexing the sum in Corollary~\ref{cor:w-space} yields \begin{align*}\mathrm{dres}\left(\sum_{i=1}^n\mathcal{L}_i(f_i),\omega,k\right) &=\sum_{i=1}^n\sum_{j=1}^{k}\lambda_{i,k-j}\frac{(-1)^{k-j}(k-1)!}{(j-1)!}\mathrm{dres}(f_i,\omega,j)\\
    &=(-1)^k(k-1)!\sum_{i=1}^n\sum_{j=1}^{k}\lambda_{i,k-j}\frac{(-1)^{j}\mathrm{dres}(f_i,\omega,j)}{(j-1)!}.\end{align*} The rest of the proof proceeds just like that of Proposition~\ref{prop:v-space}, \emph{mutatis mutandis}.
\end{proof}

\begin{rem}\label{rem:w-space-bound-range}
    Note that the range \eqref{eq:w-space-bounded} may be excessive in particular cases. If we have a given $(\mathcal{L}_1,\dots,\mathcal{L}_n)\in\K\left[\frac{d}{dx}\right]^n$ and we wish to test whether it belongs to $W(\mathbf{f})$, it is sufficient to verify \eqref{eq:w-space-bounded} only for $1\leq k\leq M$, for \[M:=\mathrm{max}\bigl\{\mathrm{ord}(\mathcal{L}_1)+m_1,\dots,\mathrm{ord}(\mathcal{L}_n)+m_n\bigr\},\] since the highest order of any pole of $\sum_{i=1}^n\mathcal{L}_i(f_i)$ is bounded above by $M$, and therefore $\mathrm{dres}(\sum_i\mathcal{L}_i(f_i),\omega,k)=0$ automatically for any $k>M$. In general, $M$ may be strictly smaller than the sum of the maxima $m+\mathrm{max}\{\beta_1,\dots,\beta_n\}$.
\end{rem}

Since $W(\mathbf{f})$ is the filtered union of the $W^\beta(\mathbf{f})$, the above Proposition~\ref{prop:w-space} will eventually succeed in computing a $\K\left[\frac{d}{dx}\right]$-basis of $W(\mathbf{f})$. But this is not useful in practice without a way to produce an upper bound $\beta$ that is guaranteed to be large enough. We accomplish this and a little more in the next result.

\begin{prop}\label{prop:w-space-bound} Let $f_1,\dots,f_n\in\Kx$ be proper and non-zero, let $m_i\in\N$ be the highest order of any pole of $f_i$, and let $m:=\mathrm{max}\{m_1,\dots,m_n\}$. Then the finite-dimensional $\K$-vector space \[W'(\mathbf{f}):=\bigl\{(\mathcal{L}_1,\dots,\mathcal{L}_n)\in W(\mathbf{f}) \ \big| \ \mathrm{ord}(\mathcal{L}_i)\leq m-m_i \ \ \text{for} \ \ 1\leq i \leq n\bigr\}\] contains a $\K\left[\frac{d}{dx}\right]$-basis of $W(\mathbf{f})$ as in \eqref{eq:w-space}.    
\end{prop}

\begin{proof}
    For $\boldsymbol{\mathcal{L}}=(\mathcal{L}_1,\dots,\mathcal{L}_n)\in W(\mathbf{f})$ as in \eqref{eq:w-space}, let $\beta_i:=\mathrm{ord}(\mathcal{L}_i)$
    and \begin{equation}\label{eq:M-bound}M(\boldsymbol{\mathcal{L}}):=\mathrm{max}\bigl\{\beta_1+m_1,\dots,\beta_n+m_n\bigr\}.\end{equation} Note that $M(\boldsymbol{\mathcal{L}})\leq m$ if and only if $\boldsymbol{\mathcal{L}}\in W'(\mathbf{f})$ already. Thus, supposing $M(\boldsymbol{\mathcal{L}})>m$, it suffices to construct $\widetilde{\boldsymbol{\mathcal{L}}}=(\widetilde{\mathcal{L}}_1,\dots,\widetilde{\mathcal{L}}_n)\in W'(\mathbf{f})$ such that \begin{equation}\label{eq:M-bound-reduced}M\left(\boldsymbol{\mathcal{L}}-\frac{d^{M(\boldsymbol{\mathcal{L}})-m}}{dx^{M(\boldsymbol{\mathcal{L}})-m}}\cdot\widetilde{\boldsymbol{\mathcal{L}}}\right)\leq M(\boldsymbol{\mathcal{L}})-1.\end{equation} To do this, let us simplify notation by writing $M:=M(\boldsymbol{\mathcal{L}})$, and write each \[\mathcal{L}_i=\sum_{j=0}^{\beta_i}\lambda_{i,j}\frac{d^j}{dx^j}.\] Then we \emph{define} the desired \begin{equation}\label{eq:L-tilde-def}\widetilde{\mathcal{L}}_i=\sum_{j=0}^{m-m_i}\tilde{\lambda}_{i,j}\frac{d^j}{dx^j}\qquad \text{by setting} \qquad
    \tilde{\lambda}_{i,j}:=\lambda_{i,M-m+j},\end{equation} with the convention that $\lambda_{i,M-m+j}:=0$ for every $0\leq j\leq m-m_i$ such that $M-m+j> \beta_i$. For any $1\leq i\leq n$ such that $\beta_i\geq M-m$, we have that $\tilde{\beta}_i:=\mathrm{ord}\left(\widetilde{\mathcal{L}}_i\right)$ is the largest $j\leq m-m_i$ such that $\lambda_{i,M-m+j}\neq 0$, which must be $\tilde{\beta}_i=\beta_i-(M-m)$, because $\beta_i+m_i\leq M$ by definition \eqref{eq:M-bound}, which implies that $\beta_i-(M-m)\leq m-m_i $. Thus, $M-m+\tilde{\beta}_i=\beta_i$ for each $1\leq i \leq n$ such that $\beta_i\geq M-m$, and 
    since $\tilde{\lambda}_{i,\tilde{\beta}_i}=\lambda_{i,\beta_i}$ for such $i$ by construction \eqref{eq:L-tilde-def}, we see that \begin{equation}\label{eq:beta-bound}\mathrm{ord}\left(\mathcal{L}_i-\frac{d^{M-m}}{dx^{M-m}}\cdot\widetilde{\mathcal{L}}_i\right)\leq\beta_i-1\end{equation}  for such $i$. We note there is at least one $i$ such that $\beta_i\geq M-m$, since $\beta_i+m_i=M$ for some $i$ and thus $\beta_i=M-m_i\geq M-m$ for such an $i$. For the remaining $1\leq i \leq n$ such that $\beta_i<M-m$, for which we have set $\widetilde{\mathcal{L}}_i=0$ in \eqref{eq:L-tilde-def}, we note that we must have $\beta_i+m_i\leq \beta_i+m< M$ by definition \eqref{eq:M-bound}. This, together with \eqref{eq:beta-bound}, implies \eqref{eq:M-bound-reduced}. Since $\tilde{\beta}_i\leq m-m_i$ by construction, it only remains to show that $\widetilde{\boldsymbol{\mathcal{L}}}\in W(\mathbf{f})$, i.e., that $\sum_{i=1}^n\widetilde{\mathcal{L}}_i(f_i)$ is summable, which by Proposition~\ref{prop:w-space} is equivalent to having \begin{equation}\label{eq:linear-system}\begin{gathered}\sum_{i=1}^n\sum_{j=1}^m\tilde{\lambda}_{i,k-j}\frac{(-1)^jD_{i,j}}{(j-1)!}=0 \\ \text{for each}\qquad 1\leq k \leq m;\end{gathered}\qquad\Longleftrightarrow\qquad\begin{gathered}\sum_{i=1}^n\sum_{j=1}^m\lambda_{i,k-j}\frac{(-1)^jD_{i,j}}{(j-1)!}=0 \\ \text{for each}\quad 1+M-n\leq k \leq M;\end{gathered}\end{equation} where the equivalence arises from the definition \eqref{eq:L-tilde-def} and the restricted ranges for the orders $k$ are because the highest order of any pole of $\widetilde{\mathcal{L}}_i(f_i)$ is $m$, and similarly the highest order of any pole of $\mathcal{L}_i(f_i)$ is at most $M$ (cf.~Remark~\ref{rem:w-space-bound-range}). By Proposition~\ref{prop:w-space}, the rightmost linear system in \eqref{eq:linear-system} is satisfied, concluding the proof.
    \end{proof}

\section{Connections with Galois theories of difference equations}\label{sec:galois}

Our interest in computing the $\K$-vector space $V(\mathbf{f})$ in \eqref{eq:v-space} and the $\K\left[\frac{d}{dx}\right]$-module $W(\mathbf{f})$ in \eqref{eq:w-space} for $\mathbf{f}=f_1,\dots,f_n\in\Kx^n$ arises from Galois theoretic-considerations. Let us explain how the algorithmic results of the previous section devolve into algorithms to compute Galois groups.

The Galois theory of linear difference equations developed in \cite{vanderput-singer:1997} associates with a system of linear difference equations over $\Kx$ a Galois group that encodes the algebraic relations among the solutions to the system. This Galois group is a linear algebraic group, that is, a group of matrices defined by a system of polynomial equations in the matrix entries. Associated with \mbox{$\mathbf{f}=(f_1,\dots,f_n)\in\Kx^n$} is the block-diagonal linear differential system \begin{equation}\label{eq:additive-system}\sigma(Y)=\begin{pmatrix}
    1 & f_1 \\ 0 & 1 \\ & &  \ddots \\ & & & 1& f_n \\ & & & 0 & 1
\end{pmatrix}Y.\end{equation} It follows from \cite[Prop.~3.1]{HardouinSinger2008} (see also \cite[Prop.~2.1]{Hardouin2008}) that the Galois group of \eqref{eq:additive-system} is the subgroup of $\mathbb{G}_a^n\simeq\overline\K^n$ consisting of $(\eta_1,\dots,\eta_n)$ such that $v_1\eta_1+\dots+v_n\eta_n=0$ for every $(v_1,\dots,v_n)\in V(\mathbf{f})$.

\begin{rem}
    We emphasize that the Galois group of \eqref{eq:additive-system} is \emph{not} $V(\mathbf{f})$ itself, as was incorrectly stated in \cite[\S8]{arreche-sitaula:2024}.
\end{rem}

Generalizing the theory of \cite{vanderput-singer:1997}, the differential Galois theory of difference equations developed in \cite{HardouinSinger2008} associates with a system of linear difference equations over $\Kx$ a differential Galois group that encodes the differential-algebraic relations among the solutions to the system. This differential Galois group is a linear differential algebraic group, that is, a group of matrices defined by a system of algebraic differential equations in the matrix entries (see \cite{cassidy:1972}, where the study of such groups was initiated). It follows from \cite[Prop.~3.1]{HardouinSinger2008} that the differential Galois group of \eqref{eq:additive-system} is the subgroup of $\mathbb{G}_a^n$ defined by the vanishing of the $n$-tuples of linear differential operators belonging to $W(\mathbf{f})$ (cf.~\cite[Prop.~11]{cassidy:1972}).

The algorithmic results of the previous section are also relevant to the Galois-theoretic study of diagonal systems \begin{equation}\label{eq:diagonal}
    \sigma(Y)=\begin{pmatrix} r_1 & & \\ & \ddots & \\ & & r_n\end{pmatrix}Y.
\end{equation} As shown in \cite[\S2.2]{vanderput-singer:1997}, the difference Galois group of \eqref{eq:diagonal} is \[\Gamma:=\left\{(\gamma_1,\dots,\gamma_n)\in\Bigl(\overline\K^\times\Bigr)^n \ \middle| \ \gamma_1^{e_1}\cdots\gamma_n^{e_n}=1 \ \text{for} \ \mathbf{e}\in E\right\},\] where $E\subseteq \Z^n$ is the subgroup of $\mathbf{e}=(e_1,\dots,e_n)$ such that \begin{equation}\label{eq:multiplicative-relation}r_1^{e_1}\cdots r_n^{e_n}=\sigma(p_\mathbf{e})/p_\mathbf{e}\end{equation} for some $p_\mathbf{e}\in\Kx^\times$. Now suppose \eqref{eq:multiplicative-relation} holds, and let us write \[f_i:=\frac{\frac{d}{dx}(r_i)}{r_i}\qquad \text{for}\quad 1\leq i\leq n;\qquad \text{and} \quad g_\mathbf{e}:=\frac{\frac{d}{dx}(p_\mathbf{e})}{p_\mathbf{e}};\] so that we have \begin{equation}\label{eq:special-multi-telescoper}e_1f_1+\dots+e_nf_n=\sigma(g_\mathbf{e})-g_\mathbf{e}.\end{equation} As in \cite[Cor.~2.1]{arreche:2017}, we observe that this version of problem \eqref{eq:multi-telescoper} is even more special because the $f_i$ have only first-order residues, and they all belong to $\Z$. So we can compute the $\mathbb{Q}$-vector space $V(\mathbf{f})$ of solutions to \eqref{eq:special-multi-telescoper} using $\mathtt{SimpleReduction}^+(f_1,\dots,f_n)$, just as in Proposition~\ref{prop:v-space}, as a preliminary step. Then we can compute a $\mathbb{Z}$-basis $\mathbf{e}_1,\dots,\mathbf{e}_s$ of the free abelian group $\tilde{E}:= V\cap\Z^n$. Since each $g_{\mathbf{e}_j}$ has only simple poles with integer residues, one can compute explicitly $p_{\mathbf{e}_j}\in\Kx$ such that $\frac{d}{dx}p_{\mathbf{e}_j}=g_{\mathbf{e}_j}p_{\mathbf{e}_j}$, and thence constants $\varepsilon_{j}\in\K^\times$ such that \[r_1^{e_{j,1}}\dots r_n^{e_{j,n}}=\varepsilon_j\frac{\sigma(p_{\mathbf{e}_j})}{p_{\mathbf{e}_j}}.\] This reduces the computation of $E$ from the defining multiplicative condition \eqref{eq:multiplicative-relation} in $\Kx^\times$ modulo the subgroup $\{\sigma(p) / p \ | \ p\in\Kx^\times\}$ to the equivalent defining condition in $\K^\times$: \[{ 
E=\left\{\sum_{j=1}^s m_j\mathbf{e}_j \in \tilde{E}\ \, \middle| \ \,\prod_{j=1}^s\gamma_j^{m_j}=1\right\}}.\]

\section{Preliminary implementations and timings}

We have developed preliminary implementations of Algorithm~\ref{alg:hermite-list}, Algorithm~\ref{alg:simple-reduction}, and Algorithm~\ref{alg:dres} in the computer algebra system Maple. These preliminary implementations are available upon email request to the first author. We plan to eventually make optimized implementations publically available.

\subsection{Implementation comments} Our implementation of $\mathtt{HermiteList}$ differs slightly from what is literally written in Algorithm~\ref{alg:hermite-list} in step 3, where our actual code sets \[(g,f_{m+1})\gets \mathtt{HermiteReduction}(g,x)\cdot\begin{pmatrix}
    (-1)^{m+1}(m+1) & 0 \\ 0 & 1
\end{pmatrix},\] which allows us to skip the computation of $m$ factorials in step 6 and return the $f_k$ directly from the loop in steps 2--5. Maple's implementation of what we have called $\mathtt{HermiteReduction}$ is contained in the more general command $\mathtt{ReduceHyperexp}$ in the package $\mathtt{DEtools}$.

Our implementation of $\mathtt{SimpleReduction}$ uses Maple's $\mathtt{autodispersion}$ command in the package $\mathtt{LREtools}$, then adds $1$ to all the operands to the output, to create the set $S$ in step 2 of Algorithm~\ref{alg:simple-reduction}. It is unsurprising, and indeed we have also verified experimentally, that this is vastly more efficient than calling on Algorithm~\ref{alg:shift-set}. In step 15 we had intended to call on Maple's conversion to partial fractions $\mathtt{parfrac}$ conversion procedure, whose programmer entry point form allows to input a numerator and list of factors of the denominator, and is supposed to return an output equivalent to \eqref{eq:parfrac-def}. This is essential to our implementation of Algorithm~\ref{alg:simple-reduction} (see steps~15--16). However, this conversion algorithm seems to have a bug: as of this writing, Maple returns an error for $\mathtt{convert}\bigl([1,[2x-1,1],[2x+1,1]],\mathtt{parfrac},x\bigr)$, explaining that the factors in the denominator must be relatively prime. In order to work around the unreliability of Maple's $\mathtt{parfrac}$ conversion, we have applied the $\mathtt{ExtendedEuclideanAlgorithm}$ command in the package $\mathtt{Algebraic}$ to compute the $a_\ell$ in step~15 of Algorithm~\ref{alg:simple-reduction} as the remainder of $a\cdot d_\ell$ modulo $b_\ell$, where $a$ is the numerator of $f$ and $d_\ell$ is the first B\'ezout coefficient in \[d_\ell\cdot\frac{b}{b_\ell}+\tilde{d}_\ell\cdot b_\ell=1.\]

\subsection{Timings} In \cite{arreche-sitaula:2024} we bemoaned the dearth of experimental evidence for the efficiency of our methods. We have now carried out some experiments with the implementations described above of Algorithm~\ref{alg:hermite-list}: $\mathtt{HermiteList}$ and Algorithm~\ref{alg:dres}: $\mathtt{DiscreteResidues}$ on Maple~19 in an up-to-date MacBook Pro with 8GB RAM on a 3.2GHz 8-core Apple M1 processor. We did not test Algorithm~\ref{alg:simple-reduction}: $\mathtt{SimpleReduction}$ individually because our implementation assumes that the input rational function is squarefree (even though this should of course be satisfied generically). We produced random rational functions with monic denominators of varying degrees $d$ and numerators of degree $d-1$ in increments of $20$ from $20$ to $500$, with the coefficients of the numerator and denominator polynomials taken randomly from the range $[-100,100]$. For each choice of degree, we computed the average execution time of $10$ runs of each algorithm. The results are in Table~\ref{tab:execution_time_HermiteList} and Table~\ref{tab:execution_time_DiscreteResidues} below.

\begin{table}
\caption{Average execution time for Algorithm~\ref{alg:hermite-list}: $\mathtt{HermiteList}$ on random rational functions of varying degrees.}
\bigskip
    \centering
    \begin{tabular}{|c|c|}
        \hline
        Degree & Average Time (seconds) \\
        \hline
        20  & 0.0215  \\
        \hline
        40  & 0.1372  \\
        \hline
        60  & 0.2091  \\
        \hline
        80  & 0.3835  \\
        \hline
        100 & 0.4908  \\
        \hline
        120 & 0.7858  \\
        \hline
        140 & 1.0239  \\
        \hline
        160 & 1.3994  \\
        \hline
        180 & 1.7390  \\
        \hline
        200 & 2.2978  \\
        \hline
        220 & 2.5288  \\
        \hline
        240 & 3.5972  \\
        \hline
        260 & 4.1288  \\
        \hline
        \end{tabular}
        \quad
        \begin{tabular}{|c|c|}
        \hline
        Degree & Average Time (seconds) \\
        \hline
        280 & 4.6270  \\
        \hline
        300 & 5.8105  \\
        \hline
        320 & 6.1173  \\
        \hline
        340 & 6.7782  \\
        \hline
        360 & 7.4908  \\
        \hline
        380 & 9.5754  \\
        \hline
        400 & 9.4504  \\
        \hline
        420 & 11.3759 \\
        \hline
        440 & 12.9153 \\
        \hline
        460 & 14.1303 \\
        \hline
        480 & 13.6151 \\
        \hline
        500 & 16.9689 \\
        \hline
        \multicolumn{1}{c}{}
    \end{tabular}
\label{tab:execution_time_HermiteList}
\end{table}

\begin{table}
\caption{Average execution time for Algorithm~\ref{alg:dres}: $\mathtt{DiscreteResidues}$ on random rational functions of varying degrees.}
\bigskip
    \centering
    \begin{tabular}{|c|c|}
        \hline
        Degree & Average Time (seconds) \\
        \hline
        20  & 0.0848  \\
        \hline
        40  & 0.1157  \\
        \hline
        60  & 0.3036  \\
        \hline
        80  & 0.6541  \\
        \hline
        100 & 1.2574  \\
        \hline
        120 & 2.2221  \\
        \hline
        140 & 3.6961  \\
        \hline
        160 & 5.5906  \\
        \hline
        180 & 8.3852 \\
        \hline
        200 & 12.6750 \\
        \hline
        220 & 17.4392 \\
        \hline
        240 & 24.0793 \\
        \hline
        260 & 31.7930 \\
        \hline
        \end{tabular}
        \quad
        \begin{tabular}{|c|c|}
        \hline
        Degree & Average Time (seconds) \\
        \hline
        280 & 42.5063 \\
        \hline
        300 & 54.2060 \\
        \hline
        320 & 67.7874 \\
        \hline
        340 & 83.4885 \\
        \hline
        360 & 101.7343 \\
        \hline
        380 & 124.0583 \\
        \hline
        400 & 151.0144 \\
        \hline
        420 & 180.0141 \\
        \hline
        440 & 210.6473 \\
        \hline
        460 & 250.6158 \\
        \hline
        480 & 287.9018 \\
        \hline
        500 & 335.6187 \\
        \hline
         \multicolumn{1}{c}{}
        \end{tabular}
\label{tab:execution_time_DiscreteResidues}
\end{table}

Although these timings suggest to us that the procedures we work very quickly on generic examples, this is probably only because generic rational functions have squarefree and shiftfree denominators already. For such rational functions, $\mathtt{HermiteList}$ (Algorithm~\ref{alg:hermite-list}) returns them unchanged, and so does Algorithm~\ref{alg:simple-reduction}: $\mathtt{SimpleReduction}$ (where we expect the bottleneck to be the computation of the eventually trivial autodispersion set), and then the computation of $\mathtt{FirstResidues}$ in step~4 of Algorithm~\ref{alg:dres}: $\mathtt{DiscreteResidues}$ should also be very fast. 

Thus a real stress test for our algorithms would come from ``generic worst-case'' examples, i.e., with denominators having high-degree irreducible factors occurring at high multiplicities with large dispersion sets. We have already carried out such tests on $\mathtt{HermiteList}$, for which the dispersion issue is irrelevant, as follows. For each \emph{seed degree} $d\in\{1,\dots,5\}$, we have produced random monic polynomials $R_1,\dots,R_{10}$ of degree $d$, again with coefficients taken randomly from the range $[-100,100]$, and then used them to create the denominator polynomial $\prod_{i=1}^{10}(R_i)^i$, which has degree $55d$, and finally made a rational function with a similarly random (no longer necessarily monic) numerator polnynomial of degree $55d-1$. For each such choice of seed degree $d$, we again computed the average execution time of $10$ runs of the $\mathtt{HermiteList}$ algorithm. The results are in the Table~\ref{tab:hermite-worst-case}. Note that in these worst-case experiments we are forcing the loop in steps 2--5 of Algorithm~\ref{alg:hermite-list} to run for 10 iterations (generically, that is, but always \emph{at least} this many times), with the $i$-th iteration computing Maple's $\mathtt{ReduceHyperexp}$ (from the package $\mathtt{DEtools}$) on a rational function of degree $d(11-i)(12-i)/2$. However, as we can begin to see in the small explicit examples discussed in the next section, the coefficients of the successive outputs in the loop seem to grow quite fast, which helps explain why the timings in Table~\eqref{tab:hermite-worst-case} are so large when compared to the baseline timings in Table~\ref{tab:execution_time_HermiteList}.

\begin{table}
\caption{Average execution time for Algorithm~\ref{alg:hermite-list}: $\mathtt{HermiteList}$ on random worst-case rational functions of varying degrees.}
\bigskip
    \centering
    \begin{tabular}{|c|c|c|}
\hline 
Seed degree & Total degree & Average time (seconds) \\
\hline
1 & 55 & 0.6857 \\
\hline
2 & 110 & 7.0305 \\
\hline
3 & 165 & 29.3660 \\ 
\hline
4 & 220 & 80.3166 \\
\hline
5 & 275 & 168.3885\\
\hline
    \end{tabular}
\label{tab:hermite-worst-case}
\end{table}

\section{Examples} Let us conclude by illustrating some of our procedures on some concrete examples. As in \cite{arreche-sitaula:2024}, we write down explicitly the irreducible factorizations of denominator polynomials in order to make it easier for the human reader to verify the correctness of the computations. We emphasize and insist upon the fact that none of our procedures performs any such factorization, except possibly within the black box computation of $\mathtt{ShiftSet}$. 

\subsection{Example} The following example is treated in \cite[Example~6.6]{Paule:1995}. Consider the rational function
\begin{equation}
    \label{eq:ex2-f}f:=\frac{x+2}{x(x^2-1)^2(x^2+2)^2}
\end{equation}
We first apply Algorithm~\ref{alg:hermite-list} to compute $\mathtt{HermiteList}(f)=:(f_1,f_2)$, which results in
\begin{align*}
    f_1&=-\frac{x^3+4x^2+13x+36}{36x(x^2-1)(x^2+2)};\\
    f_2&=\frac{x^3+2x^2+5x+10}{36(x^2-1)(x^2+2)}.
\end{align*} The rest of the algorithm works on these $f_1$ and $f_2$ components in parallel.

Let us begin with the simplest case $f_2$. Denoting its monic denominator 
\[b_2=(x^2-1)(x^2+2),\] we compute $\mathtt{ShiftSet}(b_2)=\{2\}$, and the factorization $b_2=b_{2,0}b_{2,2}$ computed within Algorithm~\ref{alg:simple-reduction} is given by \[b_{2,0}=(x+1)(x^2+2)\qquad\text{and}\qquad b_{2,2}=\mathrm{gcd}(\sigma^{-2}(b_{2,0}),b_2)=x-1.\]The individual summands in the partial fraction decomposition of $f_2$ relative to this factorization are given by \[\frac{a_{2,0}}{b_{2,0}}=-\frac{2x^2+3x+4}{36(x+1)(x^2+2)}\qquad\text{and}\qquad\frac{a_{2,2}}{b_{2,2}}=\frac{1}{12(x-1)}.\] The reduced form $\bar{f}_2=\frac{a_{2,0}}{b_{2,0}}+\sigma^2\left(\frac{a_{2,2}}{b_{2,2}}\right)$ is given by \[\bar{f}_2=\frac{x^2-3x+2}{36(x+1)(x^2+2)}.\] The pair of polynomials $(B_2,D_2)\in\K[x]\times\K[x]$ encoding the discrete residues of order $2$ of $f$ is given by \begin{equation}\label{eq:paule-ex-2}B_2=(x+1)(x^2+2)\qquad\text{and}\qquad D_2=\frac{1}{36}x^2+\frac{1}{72}x+\frac{1}{24}.\end{equation}

Let us now continue with the last case $f_1$. Denoting its monic denominator \[b_1=x(x^2-1)(x^2+2),\] we compute $\mathtt{ShiftSet}(b_1)=\{1,2\}$, and the factorization $b_1=b_{1,0}b_{1,1}b_{1,2}$ computed within Algorithm~\ref{alg:simple-reduction} is given by
\begin{gather*}b_{1,0}=(x+1)(x^2+2);\qquad b_{1,1}=\mathrm{gcd}(\sigma^{-1}(b_{1,0}),b_1)=x;\\ \text{and}\qquad b_{1,2}=\mathrm{gcd}(\sigma^{-2}(b_{1,0}),b_1)=x-1.\end{gather*} The individual summands in the partial fraction decomposition of $f_1$ relative to this factorization are given by \[\frac{a_{1,0}}{b_{1,0}}=-\frac{9x^2+x+5}{36(x+1)(x^2+2)};\qquad\frac{a_{1,1}}{b_{1,1}}=\frac{1}{2x};\qquad\text{and}\qquad\frac{a_{1,2}}{b_{1,2}}=\frac{-1}{4(x-1)}.\] The reduced form $\bar{f}_1=\frac{a_{1,0}}{b_{1,0}}+\sigma\left(\frac{a_{1,1}}{b_{1,1}}\right)+\sigma^2\left(\frac{a_{1,2}}{b_{1,2}}\right)$ is given by \[\bar{f}_1=\frac{-x+13}{36(x+1)(x^2+2)}.\] The pair of polynomials $(B_1,D_1)\in\K[x]\times\K[x]$ encoding the discrete residues of order $1$ of $f$ is given by \begin{equation}\label{eq:paule-ex-1}B_1=(x+1)(x^2+2)\qquad\text{and}\qquad D_1=-\frac{73}{1296}x^2-\frac{11}{432}x+\frac{51}{648}.\end{equation}

Let us now compare the discrete residues of $f$ according to Definition~\ref{def:dres} with the combined output \[\mathtt{DiscreteResidues}(f)=\Bigl((B_1,D_1),(B_2,D_2)\Bigr)\] from \eqref{eq:paule-ex-1} and \eqref{eq:paule-ex-2} produced by Algorithm~\ref{alg:dres}. The complete partial fraction decomposition of $f$ is given by 
\begin{multline*}\overbrace{\frac{-13/108}{x+1}+\frac{1/2}{x}+\frac{-1/4}{x-1}+\frac{(-28\pm11\sqrt{-2})/432}{x\pm\sqrt{-2}}}^{f_1}\\ +\underbrace{\frac{-1/36}{(x+1)^2}+\frac{1/12}{(x-1)^2}+\frac{(-1\mp\sqrt{-2})/72}{(x\pm\sqrt{-2})^2}}_{-\frac{d}{dx}f_2}.
\end{multline*} We see in particular that the set of poles of $f$ is \[\{-1,0,1,\sqrt{-2},-\sqrt{-2}\},\] and that each of these poles belongs to one of the three orbits \[\omega(0)=\Z;\qquad\omega(\sqrt{-2})=\sqrt{-2}+\Z;\qquad\text{and}\qquad\omega(-\sqrt{-2})=-\sqrt{-2}+\Z.\] Therefore, $f$ has no discrete residues outside of these orbits. Moreover, we verify that the set of roots $\{-1,\sqrt{-2},-\sqrt{-2}\}$ of $B_1=B_2$ contains precisely one representative of each of these orbits. Now we can compute the discrete residues of $f$ as in \eqref{eq:dres-def}:
\begin{align*}
    \mathrm{dres}(f,\omega(0),1) &=\frac{-13}{108}+\frac{1}{2} +\frac{-1}{4}=\frac{7}{54}=D_1(-1)         \\
    \mathrm{dres}(f,\omega(\sqrt{-2}),1) &=  \frac{-28-11\sqrt{-2}}{432}=D_1(\sqrt{-2})        \\
    \mathrm{dres}(f,\omega(-\sqrt{-2}),1) &=  \frac{-28+11\sqrt{-2}}{432}=D_1(-\sqrt{-2})           \\
    \mathrm{dres}(f,\omega(0),2) &=\frac{-1}{36}+\frac{1}{12} =\frac{1}{18}=D_2(-1)         \\
    \mathrm{dres}(f,\omega(\sqrt{-2}),2) &=   \frac{-1+\sqrt{-2}}{72}=D_2(\sqrt{-2})       \\
    \mathrm{dres}(f,\omega(-\sqrt{-2}),2) &=   \frac{-1-\sqrt{-2}}{72}=D_2(-\sqrt{-2}).         
\end{align*}

\subsection{Example} The following example was partially considered in \cite{arreche-sitaula:2024}, where it was incorrectly attributed  to the earlier references \cite[\S3]{Malm:1995} and \cite[\S8]{Pirastu1995b}, in which the authors consider a nearly identical rational function save for a factor of $(x+2)^2$ in the denominator in place of our factor of $(x+2)^3$. We continue to consider as in \cite{arreche-sitaula:2024} the rational function 
\begin{equation}\label{eq:ex1-f}f:=\frac{1}{x^3(x + 2)^3(x + 3)(x^2 + 1)(x^2 + 4x + 5)^2}.  \end{equation}
We first apply Algorithm~\ref{alg:hermite-list} to compute $\mathtt{HermiteList}(f)=:(f_1,f_2,f_3)$, which results in
\begin{align*}
    f_1 &= \frac{787x^5 + 4803x^4 + 9659x^3 + 9721x^2 + 9502x + 5008}{18000(x^2 + 1)(x + 3)(x^2 + 4x + 5)(x + 2)x};\\
f_2 &=-\frac{787x^3 + 3372x^2 + 4696x + 1030}{18000(x^2 + 4x + 5)x(x + 2)};\vphantom{\begin{matrix}X \\ X \\X \end{matrix}}\\
f_3&=-\frac{7x - 1}{300(x + 2)x}.
    \end{align*} The rest of the algorithm works on these $f_1$, $f_2$, and $f_3$ components in parallel. 
    
    Let us begin with the simplest case $f_3$. Denoting its monic denominator \[b_3=x(x+2),\] we compute 
    $\mathtt{ShiftSet}(b_3)=\{2\}$, and the factorization $b_3=b_{3,0}b_{3,2}$ computed within Algorithm~\ref{alg:simple-reduction} is given by \[b_{3,0}=x+2;\qquad \text{and}\qquad b_{3,2}=\mathrm{gcd}(\sigma^{-2}(b_{3,0}),b_3)=x.\] The individual summands in the partial fraction decomposition of $f_3$ with respect to this factorization (which, in this exceptionally simple case, happens to coincide with the complete partial fraction decomposition of $f_3$) are given by \[\frac{a_{3,0}}{b_{3,0}}=\frac{-1}{40(x+2)};\qquad \text{and}\qquad \frac{a_{3,2}}{b_{3,2}}=\frac{1}{600x}.\] The reduced form $\bar{f}_3=\frac{a_{3,0}}{b_{3,0}}+\sigma^2\left(\frac{a_{3,2}}{b_{3,2}}\right)$ is given by \[\bar{f}_3=\frac{-7}{300(x+2)}.\] The pair of polynomials $(B_3,D_3)\in\K[x]\times\K[x]$ encoding the discrete residues of order $3$ of $f$ is given by \begin{equation}\label{eq:ex1out3}B_3=x+2\qquad \text{and}\qquad D_3=-7/300.\end{equation}

    Let us now continue with the next case $f_2$. Denoting its monic denominator \[b_2=(x^2+4x+5)x(x+2),\] we compute $\mathtt{ShiftSet}(b_2)=\{2\}$, and the factorization $b_2=b_{2,0}b_{2,2}$ computed within Algorithm~\ref{alg:simple-reduction} is given by \[b_{2,0}=(x^2+4x+5)(x+2);\qquad \text{and}\qquad b_{2,2}=\mathrm{gcd}(\sigma^{-2}(b_{2,0}),b_2)=x.\] The individual summands in the partial fraction decomposition of $f_2$ with respect to this factorization are given by \[\frac{a_{2,0}}{b_{2,0}}=-\frac{(x^2+306x+373)}{2000(x^2+4x+5)(x+2)};\qquad \text{and}\qquad \frac{a_{2,2}}{b_{2,2}}=-\frac{103}{18000x}.\] The reduced form $\bar{f}_2=\frac{a_{2,0}}{b_{2,0}}+\sigma^2\left(\frac{a_{2,2}}{b_{2,2}}\right)$ is given by \[\bar{f}_2=-\frac{787x^2+3166x+3872}{18000(x^2+4x+5)(x+2)}.\] The pair of polynomials $(B_2,D_2)\in\K[x]\times\K[x]$ encoding the discrete residues of order $2$ of $f$ is given by \begin{equation}\label{eq:ex1out2}B_2=(x^2+4x+5)(x+2)\qquad \text{and}\qquad D_2=-\frac{1277}{36000}x^2-\frac{509}{3600}x-\frac{403}{2250}.\end{equation}

    There remains the most complicated case of $f_1$, for which we reproduce the computation from \cite{arreche-sitaula:2024} for ease of reference. Denoting by
    \[ b_1:=(x^2 + 1)(x + 3)(x^2 + 4x + 5)(x + 2)x\] the monic denominator of $f_1$, we compute $\mathtt{ShiftSet}(b_1)=\{1,2,3\}$. The factorization $b_1=b_{1,0}b_{1,1}b_{1,2}b_{1,3}$ computed within Algorithm~\ref{alg:simple-reduction} is given by
    \begin{align*}
    b_{1,0}&=(x + 3)(x^2 + 4x + 5); &
    b_{1,1} &=\mathrm{gcd}(\sigma^{-1}(b_{1,0}),b_1)=x+2;\\
    b_{1,2} & =\mathrm{gcd}(\sigma^{-2}(b_{1,0}),b_1)= x^2+1;\hphantom{XX} &
    b_{1,3} &=\mathrm{gcd}(\sigma^{-3}(b_{1,0}),b_1)=x.
\end{align*}
The individual summands in the partial fraction decomposition of $f_1$ with respect to this factorization are given by
   \begin{align*}
    \frac{a_{1,0}}{b_{1,0}}&=-\frac{13391x^2 + 37742x - 9293}{1080000(x + 3)(x^2 + 4x + 5)}; &
    \frac{a_{1,1}}{b_{1,1}} &=\frac{1}{250(x+2)}\vphantom{\begin{matrix}X\\ X \\ X \end{matrix}};\\
    \frac{a_{1,2}}{b_{1,2}} &= \frac{-7x-1}{8000(x^2+1)}; &
    \frac{a_{1,3}}{b_{1,3}} &=\frac{313}{33750x}.
\end{align*}
The reduced form $\bar{f}_1=\frac{a_{1,0}}{b_{1,0}}+\sigma\left(\frac{a_{1,1}}{b_{1,1}}\right)+\sigma^2\left(\frac{a_{1,2}}{b_{1,2}}\right)+ \sigma^3\left(\frac{a_{1,3}}{b_{1,3}}\right)$ is given by \[\bar{f}_1=\frac{273x + 1387}{20000(x + 3)(x^2 + 4x + 5)}.\] The pair of polynomials $(B_1,D_1)\in\K[x]\times\K[x]$ encoding the discrete residues of order $1$ of $f$ is given by \begin{equation}\label{eq:ex1out1}B_1=(x+3)(x^2+4x+5) \qquad \text{and} \qquad D_1=\frac{59}{16000}x^2+\frac{33}{40000}x-\frac{1321}{80000}.\end{equation}

Let us now compare the discrete residues of $f$ according to Definition~\ref{def:dres} with the combined output \[\mathtt{DiscreteResidues}(f)=\Bigl((B_1,D_1),(B_2,D_2),(B_3,D_3)\Bigr)\] from \eqref{eq:ex1out1}, \eqref{eq:ex1out2}, and \eqref{eq:ex1out3} produced by Algorithm~\ref{alg:dres}. The complete partial fraction decomposition of $f$ in \eqref{eq:ex1-f} is given by 
\begin{multline*}
    \overbrace{\frac{1/1080}{x+3}+\frac{1/250}{x+2}+\frac{313/33750}{x}+\frac{-7\mp\sqrt{-1}}{x\pm\sqrt{-1}}+\frac{-533\pm1119\sqrt{-1}}{x+2\pm\sqrt{-1}}}^{f_1}\\
    \vphantom{\begin{matrix}X\\ X \\ X\end{matrix}}+\underbrace{\frac{-13/400}{(x+2)^2}+\frac{-103/18000}{x^2}+\frac{(-11\mp2\sqrt{-1})/4000}{(x+2\pm\sqrt{-1})^2}}_{-\frac{d}{dx}f_2}
    {}+\underbrace{\frac{-1/40}{(x+2)^3}+\frac{1/600}{x^3}}_{\frac{1}{2}\frac{d^2}{dx^2}f_3}
\end{multline*}
We see in particular that the set of poles of $f$ is \[\bigl\{-3,-2,0,\sqrt{-1}-2, \sqrt{-1},-\sqrt{-1}, -\sqrt{-1}-2\bigr\},\] 
and that each of these poles belongs to one of the three orbits \[\omega(0)=\mathbb{Z};\qquad\omega\bigl( \sqrt{-1}\bigr)=\sqrt{-1}+\Z;\qquad\text{and}\qquad\omega\bigl(- \sqrt{-1}\bigr)=-\sqrt{-1}+\Z.\] Therefore, $f$ has no discrete residues outside of these orbits. We see that the single root $-2$ of the polynomial $B_3$ represents the only orbit $\omega(0)$ at which $f$ has poles of order $3$. Moreover, we verify that the set of roots $\{-2,\sqrt{-1}-2,-\sqrt{-1}-2\}$ of the polynomial $B_2$ in \eqref{eq:ex1out2} contains precisely one representative of each of the orbits at which $f$ has poles of order $2$, and that similarly so does the set of roots $\{-3,\sqrt{-1}-2,-\sqrt{-1}-2\}$ of the polynomial $B_1$ in \eqref{eq:ex1out1} contain precisely one representative of each orbit at which $f$ has poles of order $1$. And yet, the representatives $-3$ in $B_1$ and $-2$ in $B_2$ for the same orbit $\omega(0)$ do not agree (cf.~Remark~\ref{rem:dres-deficiency}). Now we compute the discrete residues of $f$ as in \eqref{eq:dres-def}:
\begin{align*}
    \mathrm{dres}\bigl(f,\omega(0),1\bigr) &= \frac{1}{1080}+\frac{1}{250}+\frac{313}{33750}=\frac{71}{5000}      \\
    \mathrm{dres}\bigl(f,\omega(\sqrt{-1}),1\bigr) &= \frac{-533-1119\sqrt{-1}}{80000} +\frac{-7+\sqrt{-1}}{16000}=\frac{-284-557\sqrt{-1}}{40000}    \\   \mathrm{dres}\bigl(f,\omega(-\sqrt{-1}),1\bigr) &= \frac{-533+1119\sqrt{-1}}{80000} +\frac{-7-\sqrt{-1}}{16000}=\frac{-284+557\sqrt{-1}}{40000}        \\   \mathrm{dres}\bigl(f,\omega(0),2\bigr) &= -\frac{13}{400}-\frac{103}{18000}=-\frac{43}{1125}      \\   \mathrm{dres}\bigl(f,\omega(\sqrt{-1}),2\bigr) &=  \frac{-11+2\sqrt{-1}}{4000}     \\   \mathrm{dres}\bigl(f,\omega(-\sqrt{-1}),2\bigr) &=  \frac{-11-2\sqrt{-1}}{4000}     \\   \mathrm{dres}\bigl(f,\omega(0),3\bigr) &= \frac{-1}{40}+\frac{1}{600}=\frac{-7}{300}      
\end{align*}
Finally, we check that these discrete residues computed from the definition indeed agree with the values of the polynomials $D_k$ at the roots of $B_k$ that represent each given orbit:
\begin{gather*}
    D_1(-3)=\tfrac{71}{5000};\qquad     
    D_1(-2\pm\sqrt{-1})=\tfrac{-284\mp557\sqrt{-1}}{40000}; \qquad         D_2(-2)=-\tfrac{43}{1125};\\        
    D_2(-2\pm\sqrt{-1}) =  \tfrac{-11\pm2\sqrt{-1}}{4000};\qquad\text{and}\qquad     D_3(-2) =\tfrac{-7}{300}.    
\end{gather*}

\begin{rem}\label{rem:ex1}
    As discussed in Remark~\ref{rem:dres-deficiency-fix}, we can apply the modified Algorithm~\ref{alg:multi-dres} to $f$, instead of Algorithm~\ref{alg:dres} as we have done above, in order to address the deficiency that the different zeros $-3$ and $-2$ of the polynomials $B_1$ and $B_2$ respectively represent the same orbit $\omega(0)=\mathbb{Z}$. The output in this case would be \[\mathtt{DiscreteResidues}^+\bigl((f)\bigr)=\Bigl(B,(\tilde{D}_{1},\tilde{D}_{2},\tilde{D}_{3})\Bigr),\] where 
    \begin{align*}
        B&=(x+3)(x^2+4x+5)=B_1;\\
        \tilde{D}_1&= \frac{59}{16000}x^2+\frac{33}{40000}x-\frac{1321}{80000}=D_1;\\
        \tilde{D}_2&=-\frac{1259}{72000}x^2-\frac{5}{72}x-\frac{6421}{72000};\\
        \tilde{D}_3&=\frac{-7}{600}x^2-\frac{7}{150}x-\frac{35}{600}.
    \end{align*} Obviously, $\tilde{D}_1=D_1$ still evaluates to the first order residue of $f$ at each orbit represented by a root of $B=B_1$, but now the new polynomials $\tilde{D}_2$ and $\tilde{D}_3$ also still enjoy the same property:
    \begin{gather*}\tilde{D}_2(-3)=\frac{-43}{1125};\qquad\tilde{D}_2(-2\pm\sqrt{-1})=\frac{-11\pm2\sqrt{-1}}{4000};\\
    \tilde{D}_3(-3)=\frac{-7}{300};\qquad\text{and}\qquad \tilde{D}_3(-2\pm\sqrt{-1})=0.\end{gather*}
\end{rem}

\end{document}